%% file: main.tex
\documentclass[letterpaper,11pt,fleqn]{article}
\usepackage[margin=1in]{geometry}
\usepackage{amsmath,amssymb,amsthm}
\usepackage{fix-cm}
\usepackage[rgb,table]{xcolor}
\usepackage{mathtools}
\usepackage{array}
\usepackage{enumitem}
\usepackage{xcolor}
\usepackage{comment}
\usepackage{subcaption}
\usepackage[hidelinks,bookmarksnumbered=true,hypertexnames=false]{hyperref}
\usepackage{algorithm,algpseudocode}
\usepackage[capitalize,sort]{cleveref}
\usepackage{tikz}
\usepackage{thm-restate}
\usepackage{multirow}
\usepackage{xfrac}
\usepackage[symbol]{footmisc}

\input{colorscheme.tex}

\input{macros.tex}

\newcommand*{\Description}[1]{}
\newcommand*{\captionAndDescr}[1]{\caption{#1}}
\newcommand*{\figScale}{1.1}
\let\shortCite\cite

\title{Improving Approximation Guarantees for Maximin Share}

\author{Hannaneh Akrami\thanks{Max Planck Institute for Informatics and Graduiertenschule Informatik, Universit\"at des Saarlandes}\\\texttt{\small hakrami@mpi-inf.mpg.de}
\and Jugal Garg\thanks{University of Illinois at Urbana-Champaign. Supported by NSF Grant CCF-1942321}\\\texttt{\small jugal@illinois.edu}
\and Eklavya Sharma\textsuperscript{\textdagger}\\\texttt{\small eklavya2@illinois.edu}
\and Setareh Taki\thanks{Grubhub, USA}\\ \texttt{\small Staki@grubhub.com}
}

\date{\empty}

\begin{document}
\maketitle
\input{abstract}
\input{intro}
\input{related-work}
\input{prelim}
\input{4n-3}

\input{hard3.tex}
\input{algo1.tex}
\input{hard1.tex}

\bibliographystyle{alphaurl}
\bibliography{bibliography}

\appendix

\input{prelim-extra-arXiv}
\input{prio-extra.tex}

\end{document}

%% file: colorscheme.tex
\def\colorschemesepia{sepia}
\def\colorschemedark{dark}
\def\colorschemelight{light}

\ifx\colorscheme\undefined
\let\colorscheme\colorschemelight
\fi

\ifx\colorscheme\colorschemelight
\colorlet{textColor}{black}
\colorlet{bgColor}{white}
\fi

\ifx\colorscheme\colorschemesepia
\definecolor{textColor}{HTML}{433423}
\definecolor{bgColor}{HTML}{fbf0da}
\fi

\ifx\colorscheme\colorschemedark
\definecolor{textColor}{HTML}{bdc1c6}
\definecolor{bgColor}{HTML}{202124}
\definecolor{textBlue}{HTML}{8ab4f8}
\definecolor{textRed}{HTML}{f9968b}
\definecolor{textGreen}{HTML}{81e681}
\definecolor{textPurple}{HTML}{c58af9}
\else
\colorlet{textBlue}{blue!50!black}
\colorlet{textRed}{red!50!black}
\colorlet{textGreen}{green!50!black}
\definecolor{textPurple}{HTML}{681da8}
\fi

\ifx\colorscheme\colorschemelight\else
\pagecolor{bgColor}
\color{textColor}
\fi

%% file: macros.tex
\newtheorem{lemma}{Lemma}
\newtheorem{theorem}{Theorem}
\newtheorem{observation}{Observation}
\newtheorem{corollary}{Corollary}
\newtheorem{definition}{Definition}
\newtheorem{claim}{Claim}
\newtheorem*{claim*}{Claim}
\newtheorem{example}{Example}
\newtheorem{proposition}{Proposition}

\newcommand{\defeq}{:=}
\newcommand{\Th}{^{\mathrm{th}}}
\newcommand{\wLoG}{without loss of generality}
\newcommand*{\eps}{\varepsilon}
\newcommand*{\floor}[1]{\lfloor #1 \rfloor}
\newcommand*{\ceil}[1]{\lceil #1 \rceil}
\renewcommand{\t}[1]{\texttt{#1}}

\newcommand{\MMS}{\text{MMS}}
\newcommand{\I}{\mathcal{I}}

\renewcommand{\t}[1]{\texttt{#1}}

\newcommand{\ins}{\I=(N,M,V)}
\newcommand{\tf}[2]{\tfrac{#1}{#2}}

\newcommand*{\etal}{{\em et~al.}}
\newenvironment{claimproof}[1]{\emph{Proof. }\space#1}{\hfill $\blacksquare$\medskip\par}

\algnewcommand{\LineComment}[1]{\State \textcolor{gray}{\texttt{//} \textit{#1}}}

\newcommand*{\Null}{\mathtt{null}}
\DeclareMathOperator{\bagFill}{\texttt{bagFill}}
\DeclareMathOperator{\bagFillHyp}{\hyperref[algo:bagFill]{\bagFill}}
\DeclareMathOperator{\RBF}{\mathtt{RBF}}
\DeclareMathOperator{\ordSt}{\#}

\newcommand*{\Acal}{\mathcal{A}}

\newcommand*{\Ical}{\mathcal{I}}
\newcommand*{\Vcal}{\mathcal{V}}

\newlength{\cellW}\newlength{\cellH}

%% file: abstract.tex
\begin{abstract}
We consider fair division of a set of indivisible goods among $n$ agents with additive valuations using the fairness notion of maximin share (MMS). MMS is the most popular share-based notion, in which an agent finds an allocation fair to her if she receives goods worth at least her ($1$-out-of-$n$) MMS value. An allocation is called MMS if all agents receive their MMS values. However, since MMS allocations do not always exist~\texorpdfstring{\cite{KurokawaPW18}}{[Kurokawa et al., JACM'18]}, the focus shifted to investigating its ordinal and multiplicative approximations.

In the ordinal approximation, the goal is to show the existence of $1$-out-of-$d$ MMS allocations (for the smallest possible $d>n$). A series of works led to the state-of-the-art factor of $d=\lfloor3n/2\rfloor$~%
\texorpdfstring{\cite{Hosseini2021OrdinalMS}}{[Hosseini et al., J.Artif.Intell.Res.'21]}.
We show that $1$-out-of-$4\lceil n/3\rceil$ MMS allocations always exist, thereby improving the state-of-the-art of ordinal approximation. 

In the multiplicative approximation, the goal is to show the existence of $\alpha$-MMS allocations (for the largest possible $\alpha < 1$), which guarantees each agent at least $\alpha$ times her MMS value. A series of works in the last decade led to the state-of-the-art factor of $\alpha = \frac{3}{4} + \frac{3}{3836}$~%
\texorpdfstring{\cite{Akrami2023BreakingT}}{[Akrami and Garg, SODA'24]}.

We introduce a general framework of \emph{approximate MMS with agent priority ranking}. We order the agents, and agents earlier in the order are considered more \emph{important}.  An allocation is said to be $T$-MMS, for a non-increasing sequence $T \defeq (\tau_1, \ldots, \tau_n)$ of numbers, if the agent at rank $i$ in the order gets a bundle of value at least $\tau_i$ times her MMS value. This framework captures both ordinal approximation and multiplicative approximation as special cases. We show the existence of $T$-MMS allocations where $\tau_i \ge \max(\frac{3}{4} + \frac{1}{12n}, \frac{2n}{2n+i-1})$ for all $i$. Furthermore, by ordering the agents randomly, we can get allocations that are $(\frac{3}{4} + \frac{1}{12n})$-MMS ex-post and $(0.8253 + \frac{1}{36n})$-MMS ex-ante. We also investigate the limitations of our algorithm and show that it does not give better than $(0.8631 + \frac{1}{2n})$-MMS ex-ante.
\end{abstract}

%% file: intro.tex
\section{Introduction}

Fair allocation of resources (goods) is a fundamental problem in multiple disciples, including computer science, economics, and social choice theory, where the goal is to divide a set $M$ of $m$ goods among a set $N$ of $n$ agents in a \emph{fair} manner. This field has received significant attention since the seminal work of Steinhaus in the 1940s~\cite{steinhaus1948problem}.
When the goods are divisible, the two standard fairness notions are \emph{envy-freeness} and \emph{proportionality}, based on envy and share, respectively. In an envy-free allocation, no agent prefers another agent's allocation, and in a proportional allocation, each agent receives her proportionate share, i.e., at least a $1/n$ fraction of her value of all the goods. In the case of divisible goods, an envy-free and proportional allocation exists; see~\cite{Varian74,Foley67,AzizM16b}.

We study the discrete setting, in which each good can be given to exactly one agent. 
Let $v_i: 2^M \rightarrow \mathbb{R}_{\geq 0}$ denote the valuation function of agent $i$. We assume that $v_i(.)$'s are additive, i.e., $v_i(S) = \sum_{g\in S}v_i(\{g\})$. 
Simple examples show that in the discrete case, neither envy-freeness nor proportionality can be guaranteed.\footnote{Consider an instance with two agents and one good with positive value to both agents.} This necessitates the refinement of these notions.

We consider the natural and most popular discrete analog of proportionality called \emph{maximin share} (MMS) introduced in~\cite{budish2011combinatorial}. It is also shown to be preferred by participating agents in real-life experiments~\cite{GatesGD20}. The MMS value of an agent is the maximum value she can guarantee if she divides the goods into $n$ bundles (one for each agent) and then receives a bundle with the minimum value. Formally, for a set $S$ of goods and an integer $d$, let $\Pi_d(S)$ denote the set of all partitions of $S$ into $d$ bundles. Then,
\begin{align*}
    \MMS_i^d(S) := \max_{(P_1, \ldots, P_d) \in \Pi_d(S)} \min_{j} v_i(P_j).
\end{align*}
The MMS value of agent $i$ is denoted by $\MMS_i := \MMS_i^n(M)$.
An allocation is said to be MMS if each agent receives at least their MMS value. However, MMS is an unfeasible share guarantee that cannot always be satisfied when there are more than two agents with additive valuations~\cite{procaccia2014fair, KurokawaPW18, FeigeST21}. Therefore, the MMS share guarantee needs to be relaxed, and the two natural ways are its multiplicative and ordinal approximations.

\paragraph{\bf \boldmath $\alpha$-MMS}
Since we need to lower the share threshold, a natural way is to consider $\alpha<1$ times the MMS value. Formally, an allocation $X = (X_1, \ldots, X_n)$ is $\alpha$-MMS if for each agent $i$, $v_i(X_i) \geq \alpha \cdot \MMS_i$. Earlier works showed the existence of $2/3$-MMS allocations using several different approaches~\cite{procaccia2014fair,amanatidis2017approximation,barman2020approximation,garg2019approximating,KurokawaPW18}.
Later, in a groundbreaking work~\cite{ghodsi2018fair}, the existence of $3/4$-MMS allocations was obtained through more sophisticated techniques and involved analysis. This factor was slightly improved to $3/4 + 1/(12n)$ in~\cite{garg2021improved}, then more recently to $\frac{3}{4} + \min(\frac{1}{36}, \frac{3}{16n-4})$~\cite{akrami2023simplification}, and finally to $3/4+3/3836$~\cite{Akrami2023BreakingT}. On the other hand, $\alpha$-MMS allocations need not exist for $\alpha> 1-1/n^4$~\cite{FeigeST21}.

\paragraph{\bf \boldmath $1$-out-of-$d$ MMS}
Another natural way of relaxing MMS is to consider the share value of $\MMS_i^d(M)$ for $d>n$ for each agent $i$, which is the maximum value that $i$ can guarantee if she divides the goods into $d$ bundles and then takes a bundle with the minimum value. This notion was introduced together with the MMS notion in~\cite{budish2011combinatorial}, which also shows the existence of $1$-out-of-$(n+1)$ MMS after \emph{adding excess goods}. Unlike $\alpha$-MMS, this notion is robust to small perturbations in the values of goods because it only depends on the bundles' ordinal ranking and is not affected by small perturbations as long as the ordinal ranking of the bundles does not change.%
\footnote{The $\alpha$-MMS is very sensitive to agents' precise cardinal valuations: Consider the example mentioned in \cite{Hosseini2021OrdinalMS}. Assume $n=3$ and there are four goods $g_1$, $g_2$, $g_3$ and $g_4$ with values $30$, $39$, $40$ and $41$ respectively for agent $1$. Assume the goal is to guarantee the $3/4$-MMS value of each agent. We have $\MMS_1 = 40$, and therefore any non-empty bundle satisfies $3/4$-MMS for agent $1$. However, if the value of $g_3$ gets slightly perturbed and becomes $40+\epsilon$ for any $\epsilon>0$, then $\MMS_1 > 40$ and then $3/4 \cdot \MMS_1 > 30$ and the bundle $\{g_1\}$ does not satisfy agent $1$. Thus,  the acceptability of a  bundle  (in this example, $\{g_1\}$) might be affected by an arbitrarily small perturbation in the value of an irrelevant good (i.e., $g_3$).

Observe that in this example, whether the value of $g_3$ is $40$ or $40+\epsilon$ for any $\epsilon \in \mathbb{R}$, $\{g_1\}$ is an acceptable $1$-out-of-$4$ MMS bundle for agent $1$.}

In the standard setting (i.e., without excess goods), the first non-trivial ordinal approximation was the existence of $1$-out-of-$(2n-2)$ MMS allocations~\cite{AignerHorev2019EnvyfreeMI}, which was later improved to $1$-out-of-$\lceil 3n/2 \rceil$~\cite{hosseini2021guaranteeing}, and then to the current state-of-the-art $1$-out-of-$\lfloor 3n/2 \rfloor$~\cite{Hosseini2021OrdinalMS}. On the other hand, the (non-)existence of $1$-out-of-$(n+1)$ MMS allocations is open to date.

Our first main result in the following theorem shows that $1$-out-of-$4\lceil n/3\rceil$ MMS allocations always exist, thereby improving the state-of-the-art of $1$-out-of-$d$ MMS.

\begin{theorem}\label{thm:1}
$1$-out-of-$4\lceil n/3 \rceil$ MMS allocations always exist. 
\end{theorem}

Our proof is constructive; we give a simple algorithm to achieve $1$-out-of-$4\lceil n/3 \rceil$ MMS allocations. Our algorithm just utilizes \emph{bag-filling}, where we initialize $n$ bags $\{B_1, \dots, B_n\}$ with a particular set of two goods in each. 
Then, we iterate over the bags, and in each round $j$, we keep adding more goods to the bag $B_j$ until its value is at least $\MMS^d_i$ for some agent $i$, who has not received a bag yet, and $d=4\ceil{\frac{n}{3}}$. Then, we allocate $B_j$ to any such agent $i$ and proceed.

We note that our algorithm is very similar to the bag-filling phase of algorithms in~\cite{garg2021improved,akrami2023simplification} which obtain $\frac{3}{4}+O(\frac1n)$-MMS allocations. However, these algorithms use another critical phase called \emph{valid reductions} before bag-filling. Valid reductions allocate a set of goods to an agent $i$ such that while $i$ gets her desired share ($\alpha\cdot \MMS_i$), the MMS values of the other agents do not drop. The bag-filling phase in these algorithms is run only when no more valid reductions are feasible. It provides strong bounds on the values of the remaining goods utilized in the analysis of the bag-filling phase. However, it is unclear how to use valid reduction in an algorithm to find $1$-out-of-$d$ MMS allocations. This is because $d$ is not the same as $n$, and our goal is to guarantee exact $1$-out-of-$d$ MMS allocations and not an $\alpha$ approximation of it. Consequently, our analysis is totally different from the ones in ~\cite{garg2021improved,akrami2023simplification}, and much more involved and challenging. We give an overview of our analysis in~\cref{sec:tec}.
%
We also give a tight example showing that Theorem~\ref{thm:1} has the best factor possible with this approach (Section~\ref{sec:tighteg}). 

Another way to interpret $1$-out-of-$d$ MMS allocations is giving $n/d$ fraction of agents their MMS value and nothing to the remaining agents. We prove this equivalence in Section~\ref{sec:prelim}.
Both ordinal and multiplicative approximations focus on extremes. In $1$-out-of-$d$ MMS, some agents get nothing, and others are guaranteed their \emph{full} MMS value.  In $\alpha$-MMS, each agent receives (the same factor) $\alpha<1$ fraction of their MMS value.  As a middle ground between these two extreme notions, we introduce a general framework of \emph{approximate MMS with agent priority ranking}, where each agent has a unique \emph{priority rank} from the set $\{1, \ldots, n\}$. We are also given a list of \emph{thresholds} $T \defeq (\tau_1, \tau_2, \ldots, \tau_n)$ where $1 \ge \tau_1 \ge \ldots \ge \tau_n \ge 0$. An allocation is said to be $T$-MMS if the agent with priority rank $i$ gets a bundle of value at least $\tau_i$ times her MMS value. This framework captures both ordinal and multiplicative approximations as special cases. Namely, an $\alpha$-MMS allocation corresponds to the case where $\tau_i = \alpha$ for all $i \in [n]$, and $1$-out-of-$d$ allocations correspond to the case where the first $n/d$ thresholds in $T$ are 1 and the rest are 0. Furthermore, the $(\alpha, \beta)$-framework introduced in~\cite{hosseini2021guaranteeing}, where $\alpha$ fraction of agents receive $\beta$-MMS, is another special case where the first $n\alpha$ thresholds in $T$ are $\beta$ and the rest are 0.

Our second main result in the following theorem shows that a $T$-MMS allocation always exists where $\tau_i = \max(\frac{3}{4} + \frac{1}{12n}, \frac{2n}{2n+i-1})$ for all $i \in [n]$. 

\begin{theorem}\label{thm:2}
$T=(\tau_1, \ldots, \tau_n)$-MMS allocations exist when $\tau_i \defeq \max\left(\frac{2n}{2n+i-1}, \frac{3}{4} + \frac{1}{12n}\right)$ for all $i\in N$. 
\end{theorem}


We prove Theorem~\ref{thm:2} constructively by giving an allocation algorithm. Our algorithm is almost identical to that of \cite{garg2021improved}. The main difference is that when multiple agents are eligible to receive a bundle, we give the bundle to the agent with the smallest priority rank.
However, despite a similar algorithm, our analysis is totally different from theirs. Our algorithm, like theirs, performs a few \emph{reduction} operations (phase 1) until the instance becomes \emph{irreducible}, and then does bag filling (phase 2). They prove that reduction operations in their algorithm are \emph{valid}, so they only need to prove their claim for irreducible instances. In our setting, reduction operations may not be valid, so we need to prove our claim for non-irreducible instances too. To do this, we bound the values of bundles from both phases together.

Another approach to get fair allocations is to use randomization.
For example, when there is a single good, giving the good to an agent selected uniformly randomly is fair. Hence, one can also consider the fairness of probability distributions of allocations.
However, randomness by itself is unsatisfactory: giving all the goods to a random agent is fair because each agent has equal opportunity, but is unfair after the randomness is realized due to the large disparity in this allocation.
To fix this, one can aim for a \emph{best-of-both-worlds} approach
(see, e.g.,~\cite{AleksandrovAGW15,freeman2020best,babaioff2022best,aziz2023best,cycle-breaking})
i.e., find a distribution of allocations where each allocation in the support of the distribution is fair (ex-post fairness), and the entire distribution is fair in a randomized sense (ex-ante fairness).
For example, Babaioff \etal{} \shortCite{babaioff2022best} gave an algorithm whose output is ex-post $1/2$-MMS and ex-ante proportional, i.e., it outputs a distribution over $1/2$-MMS allocations, such that each agent's expected value of her bundle is at least $v_i(M)/n$.

By ordering the agents randomly in \cref{thm:2}, we can get allocations that are $(\frac{3}{4} + \frac{1}{12n})$-MMS ex-post and $(0.8253 + \frac{1}{36n})$-MMS ex-ante. On the other hand, we show that our algorithm cannot give better than ex-ante $(0.8631+\frac{1}{2n})$-MMS, regardless of how we pick the thresholds (\cref{sec:hard1}).
The analysis of many MMS-approximation algorithms has a natural property called \emph{obliviousness}. In \cref{sec:hard2}, we show that obtaining better than ex-ante $(0.8578+\frac{1}{3n})$-MMS with our algorithm (for some choice of thresholds) requires non-oblivious proof techniques.

%% file: related-work.tex
\subsection{Further Related Work}
Since the MMS notion and its variants have been intensively studied, we mainly focus on closely related work here. Computing the MMS value of an agent is NP-hard, but a PTAS exists \cite{woeginger1997polynomial}. For $n=2$, MMS allocations always exist \cite{bouveret2016characterizing}. For $n=3$, a series of work has improved the MMS approximation from $3/4$ \cite{procaccia2014fair} to $7/8$ \cite{amanatidis2017approximation} to $8/9$ \cite{gourves2019maximin}, and then to $11/12$~\cite{feige2022improved}.
For $n=4$, $(4/5)$-MMS allocations exist~\cite{ghodsi2018fair,babaioff2022fair}.

Babaioff \etal\ \cite{BabaioffNT21} considered $\ell$-out-of-$d$ MMS, in which the MMS value of an agent is the maximum value that can be guaranteed by partitioning goods into $d$ bundles and selecting the $\ell$ least-valuable ones. This was further studied by Segal-Halevi \cite{SegalHalevi2019TheMS, SegalHalevi2017CompetitiveEF}. Currently, the best result is the existence of $\ell$-out-of-$\lfloor (\ell+\frac{1}{2})n \rfloor$ MMS \cite{Hosseini2021OrdinalMS}. The MMS and its ordinal approximations have also been applied in the context of cake-cutting problems~\cite{Elkind2021GraphicalCC, Elkind2021KeepYD, Elkind2020MindTG, BogomolnaiaM22}.

Many works have analyzed randomly generated instances. Bouveret and Lema{\^\i}tre~\shortCite{bouveret2016characterizing}
showed that MMS allocations usually exist (for data generated randomly using uniform or Gaussian valuations).
MMS allocations exist with high probability when the valuation of each good is drawn independently and randomly from the uniform distribution on $[0, 1]$ \cite{amanatidis2017approximation} or for arbitrary distributions of sufficiently large variance \cite{kurokawa2016can}.

MMS can be analogously defined for fair division of chores where items provide negative value. Like the case of the goods, MMS allocations do not always exist for chores \cite{aziz2017algorithms}. Many papers studied approximate MMS for chores~\cite{aziz2017algorithms,barman2020approximation,huang2021algorithmic}, with the current best approximation ratio being $13/11$~\cite{huang2023reduction}. For three agents, $19/18$-MMS allocations exist~\cite{feige2022improved}. Also, ordinal MMS approximation for chores has been studied, and it is known that $1$-out-of-$\lfloor 3n/4 \rfloor$ MMS allocations exist~\cite{Hosseini2022OrdinalMS}. The chores case turns out to be easier than goods due to its close relation with the well-studied variants of bin-packing and job scheduling problems.

MMS has also been studied for non-additive valuations~\cite{MMS-XOS,ghodsi2018fair,li2021fair}. Generalizations have been studied where restrictions are imposed on the set of allowed allocations, like matroid constraints \cite{gourves2019maximin}, cardinality constraints \cite{biswas2018fair}, and graph connectivity constraints \cite{bei2022price,truszczynski2020maximin}. Stretegyproof versions of fair division have also been studied~\cite{barman2019fair,amanatidis2016truthful,amanatidis2017truthful,aziz2019Strategyproof}. MMS has also inspired other notions of fairness, like weighted MMS \cite{farhadi2019fair}, AnyPrice Share (APS) \cite{babaioff2021fair}, Groupwise MMS \cite{barman2018groupwise,chaudhury2021little}, and self-maximizing shares \cite{babaioff2022fair}.

%% file: prelim.tex
\section{Preliminaries}\label{sec:prelim}

For $t \in \mathbb{N}$, we denote the set $\{1, 2, \ldots, t\}$ by $[t]$. A discrete fair division instance is denoted by $(N, M, \mathcal{V})$, where $N$ is the set of $n$ agents, $M$ is the set of $m$ indivisible goods, and $\mathcal{V}=(v_1, \ldots, v_n)$ is the vector of agents' valuation function. Often, we assume (\wLoG) that $N = [n]$ and $M = [m]$. For all $i \in [n]$, $v_i: 2^M \rightarrow \mathbb{R}_{\geq 0}$ is the valuation function of agent $i$ over all subsets of the goods. In this paper, we assume $v_i(\cdot)$ is additive for all $i \in [n]$, i.e., $v_i(S)=\sum_{g \in S} v_i(\{g\})$ for all $S \subseteq M$. For ease of notation, we also use $v_i(g)$ and $v_{i,g}$ instead of $v_i(\{g\})$. An allocation $X=(X_1, \ldots, X_n)$ is a partition of the goods into $n$ bundles such that each agent $i \in [n]$ receives $X_i$.

For a set $S$ of goods and any positive integers $d$, let $\Pi_d(S)$ denote the set of all partitions of $S$ into $d$ bundles. Then,
\begin{align}
    \MMS_i^d(S) := \max_{P \in \Pi_d(S)} \min_{j=1}^d v_i(P_j). \label{MMS-def}
\end{align}

Setting $d=n$, we obtain the standard MMS notion. Formally, $\MMS_i = \MMS^n_i(M)$.
We call an allocation $X$ \emph{$1$-out-of-$d$ MMS}, if for all agents $i$, $v_i(X_i) \geq \MMS^d_i(M)$.
For each agent $i$, $d$-MMS partition of $i$ is a partition $P=(P_1, \ldots, P_d)$ of $M$ into $d$ bundles such that $\min_{j=1}^d v_i(P_j)$ is maximized. Basically, for a $d$-MMS partition $P$ of agent $i$, $\MMS_i^d(M) = \min_{j=1}^d v_i(P_j)$.
In the rest of the paper, for each agent $i$, we denote a $d$-MMS partition of $i$ by $P^i = (P^i_1, \ldots, P^i_d)$.

Consider a fair division instance with $n$ agents,
where each agent $i$ has a unique priority rank $r_i \in [n]$.
Let $T \defeq (\tau_1, \ldots, \tau_n)$ be a list of numbers,
where $1 \ge \tau_1 \ge \ldots \ge \tau_n \ge 0$.
An allocation $X$ is $T$-MMS if for all $i \in [n]$, we have $v_i(X_i) \ge \tau_{r_i} \cdot \MMS_i$.
Henceforth, unless stated otherwise, we assume that $r_i = i$ for all $i \in [n]$.
This assumption is \wLoG{}, since we can just renumber the agents.

\begin{lemma}
For integers $d \geq n$, let $T$ be a list of length $d$ having $n$ ones and $d-n$ zeros.
Then $1$-out-of-$d$ MMS allocations exist for all instances with $n$ agents and additive valuations
if and only if $T$-MMS allocations exist for all instances with $d$ agents and additive valuations.
\end{lemma}
\begin{proof}
All valuations considered in this proof are additive.

First, assume $1$-out-of-$d$ MMS allocations exist for all instances with $n$ agents.
Consider any instance $\mathcal{I} = ([d], M, (v_1, \ldots, v_d))$ with $d$ agents.
Then for the instance $\mathcal{I}' = ([n], M, (v_1, \ldots, v_n))$, a $1$-out-of-$d$ MMS allocation $X'$ exists.
Hence, $v_i(X_i') \ge \MMS_i^d(M)$ for all $i \in [n]$.
Let $X$ be an allocation of $M$ over $d$ agents where $X_i = X_i'$ for $i \in [n]$ and $X_i = \emptyset$ for $i \in [d] \setminus [n]$. Then $X$ is a $T$-MMS allocation for $\mathcal{I}$.
Since the choice of $\mathcal{I}$ was arbitrary, we get that $T$-MMS allocations exist for all instances with $d$ agents.

Now assume $T$-MMS allocations exist for all instances with $d$ agents.
Consider any instance $\mathcal{I} = ([n], M, (v_1, \ldots, v_n))$.
Add $d-n$ dummy agents with arbitrary additive valuations $v_{n+1}, \ldots, v_d$.
Let $\mathcal{I}' = ([d], M, (v_1, \ldots, v_d))$ be the resulting instance.
Then a $T$-MMS allocation $X'$ exists for $\mathcal{I}'$, i.e., $v_i(X_i') \ge \MMS_i^d(M)$ for all $i \in [n]$.
Hence, $X = (X_1', \ldots, X_n')$ is a 1-out-of-$d$ MMS allocation for $\mathcal{I}$.
Since the choice of $\mathcal{I}$ was arbitrary, we get that 1-out-of-$d$ MMS allocations exist for all instances with $n$ agents.
\end{proof}

\begin{definition}
\label{defn:ordered}
    An instance $\mathcal{I}=(N, M, \mathcal{V})$ is \emph{ordered} if there exists an ordering $[g_1, \ldots, g_m]$ of the goods such that for all agents $i$, $v_i(g_1) \geq \ldots \geq v_i(g_m)$.
\end{definition}

\begin{definition}
\label{defn:normalized}
    An instance $\mathcal{I}=(N, M, \mathcal{V})$ is $d$-normalized if for all agents $i$, there exists a partition $P = (P_1, \ldots, P_d)$ of $M$ into $d$ bundles such that $v_i(P_j)=1$ for all $j \in [d]$.
\end{definition}

Barman and Krishnamurthy \cite{barman2020approximation} proved that when the goal is to guarantee a minimum threshold of $\alpha_i$ for each agent $i$, it is without loss of generality to assume the instance is ordered.
Akrami et al. \cite{akrami2023simplification} proved that when the goal is to find an approximate MMS allocation, it is without loss of generality to assume the instance is $n$-normalized and ordered. Their proof does not rely on the number of agents. Formally, for any $d \in \mathbb{N}$, when the goal is to find a $1$-out-of-$d$ MMS allocation, it is without loss of generality to assume the instance is ordered and $d$-normalized, as shown in the following lemma
(proof is in \cref{sec:prelims-extra}).

\begin{restatable}{lemma}{ordNorm}\label{ord-norm}
    For any $d \in \mathbb{N}$, if $1$-out-of-$d$ MMS allocations exist for $d$-normalized ordered instances, then $1$-out-of-$d$ MMS allocations exist for all instances.
\end{restatable}

From now on, even if not mentioned, we assume the instance is ordered and $d$-normalized. Without loss of generality, for all $i \in [n]$, we assume $v_i(1) \geq v_i(2) \geq \ldots \geq v_i(m)$. In Section \ref{prelim-1}, we prove some properties of ordered $d$-normalized instances for arbitrary $d$. In Section \ref{sec:4n-3}, we set $d=4\ceil{n/3}$ and prove $1$-out-of-$4\ceil{n/3}$ MMS allocations always exist.

\subsection{1-out-of-d MMS}\label{prelim-1}

We prove the existence of $1$-out-of-$4\ceil{n/3}$ MMS assuming that $n$ is a multiple of 3.
This is \wLoG{}, because otherwise we can copy one of the agents $1$ or $2$ times
(depending on $n \bmod 3$) so that the new instance has $n' \defeq 3\ceil{n/3}$ agents.
Since we prove the existence of $1$-out-of-$4n'/3$ MMS for the new instance,
we prove the existence of an allocation that gives all the agents $i$ in the original instance
their $\MMS^{4n'/3}_i(M) = \MMS^{4\lceil n/3 \rceil}_i(M)$ value.
Hence, the existence of $1$-out-of-$4\lceil n/3 \rceil$ MMS allocations follows.

Recall that for a given instance $\mathcal{I}$ and integer $d$, for each agent $i$, $P^i = (P^i_1, \ldots, P^i_d)$ is a $d$-MMS partition of agent $i$.
\begin{proposition}\label{prop:trivial}
    Given a $d$-normalized instance for all $i \in N$ and $k \in [d]$, we have
    \begin{enumerate}
        \item $v_i(P^i_k)=1$, and
        \item $v_i(M)=d$.
    \end{enumerate}
\end{proposition}
We note that it is without loss of generality to assume $m \geq 2d$. Otherwise, we can add $2d-m$ dummy goods with a value of $0$ for all the agents. The normalized and ordered properties of the instance would be preserved.
Consider the bag setting with $d$ bags as follow.
\begin{equation}
    \label{eq:C_i}
    C_k := \{k , 2d-k+1\} \text{ for } k\in [d]
\end{equation}
See Figure \ref{c-bags} for more intuition. Next, we show some important properties of the values of the goods in $C_k$'s.
\begin{figure}[t]
    \centering
    \input{figs/C-bags.tikz}
    \captionAndDescr{For all $k \in [d]$, we define $C_k := \{k, 2d-k+1\}$.}
    \label{c-bags}
\end{figure}
\begin{proposition}\label{prop:exact}
    For all agents $i \in N$, we have
    \begin{enumerate}
        \item $v_i(1) \leq 1$, \label{exact:1}
        \item $v_i(C_d) \leq 1$, and \label{exact:2}
        \item $v_i(d+1) \leq \tf12$. \label{exact:3}
    \end{enumerate}
\end{proposition}
\begin{proof}
For the first part,
    fix an agent $i$. Let $1 \in P^i_1$. By Proposition \ref{prop:trivial}, $v_i(1) \leq v_i(P^i_1)=1$.

    For the second part, by the pigeonhole principle, there exists a bundle $P^i_k$ and two goods $j, j' \in \{1, 2, \ldots, d+1\}$ such that $\{j,j'\} \subseteq P^i_k$. Without loss of generality, assume $j < j'$. We have
    \begin{align*}
        v_i(C_d) &= v_i(d) + v_i(d+1) \tag{$C_d = \{d, d+1\}$}\\
        &\leq v_i(j) + v_i(j') &\tag{$j \leq d$ and $j' \leq d+1$} \\
        &\leq v_i(P^i_k) =1. &\tag{$\{j,j'\} \in P^i_k$}
    \end{align*}

    For the third part, we have
    \begin{align*}
        1 &\geq v_i(C_d) = v_i(d) + v_i(d+1) \geq 2 v_i(d+1).
    \end{align*}
    Thus, $v_i(d+1) \leq \tf12$.
\end{proof}
\begin{lemma}
\label{lem:Csum}
For all $i \in N$ and $k \in [d]$, $\sum_{j=k}^d v_i(C_j) \le d-k+1$.
\end{lemma}
\begin{proof}
For the sake of contradiction, assume the claim does not hold for some agent $i$ and let $\ell \geq 1$ be the largest index for which we have $\sum_{j = \ell}^d v_i(C_j) > d - \ell + 1$. Proposition \ref{prop:exact}(\ref{exact:2}) implies that $\ell < d$.
We have
\begin{align*}
    v_i(\ell) + v_i(2d-\ell+1) &= v_i(C_\ell) \\
    &= \sum_{j=\ell}^d v_i(C_j) - \sum_{j=\ell+1}^d v_i(C_j) \\
    &>  (d - \ell + 1) - (d - (\ell+1) + 1) &\tag{$\sum_{j = k}^d v_i(C_j) \leq d - k + 1$ for $k>\ell$} \\
    &=1.
\end{align*}
 For all $j,j' < \ell$, $v_i(j) + v_i(j') \geq v_i(\ell) + v_i(2d-\ell+1) > 1$. Therefore, $j$ and $j'$ cannot be in the same bundle in any $d$-MMS partition of $i$. For $j<\ell$, let $j \in P^i_j$. For all $j<\ell$ and $\ell \leq j' \leq 2d-\ell+1$,
 \begin{align*}
    v_i(j) + v_i(j') &\geq v_i(\ell) + v_i(2d-\ell+1) \\
    &= v_i(C_{\ell}) >1.
 \end{align*}
 Therefore, $j' \notin P_j$. Also, since $\sum_{j = \ell}^d v_i(C_j) > d - \ell + 1$, there are at least $t \geq d-\ell+2$ different bundles $Q_1, \ldots, Q_t$ in $P$ such that $Q_j \cap \{\ell, \ldots, 2d-\ell+1\} \neq \emptyset$. It is a contradiction since these $t \geq d-\ell+2$ bundles must be different from $P^i_1, \ldots P^i_{\ell-1}$.
\end{proof}

%% file: figs/C-bags.tikz
\begin{tikzpicture}
[scale=1,
 good/.style={circle, draw=black, thick, minimum size=30pt},
]

\draw[black, very thick] (-0.4-0.25,0.8) rectangle (0.4+0.25,3.4);

\node[good]      at (0,2.75)      {$\scriptstyle{2d}$};
\node[good]      at (0,1.5)      {$\scriptstyle{1}$};

\node at (0, 0.5) {$\scriptstyle{C_1}$};

\filldraw[color=black!60, fill=black!5, thick](1.6, 2) circle (0.02);
\filldraw[color=black!60, fill=black!5, thick](1.7, 2) circle (0.02);
\filldraw[color=black!60, fill=black!5, thick](1.8, 2) circle (0.02);

\draw[black, very thick] (-1.4-0.25 +4.5,0.8) rectangle (-0.6+0.25 +4.5,3.4);

\node[good, scale=0.75]      at (0 +3.5,2.75)      {$\scriptstyle{2d-k+1}$};
\node[good]      at (0 +3.5,1.5)      {$\scriptstyle{k}$};

\node at (0 +3.5, 0.5) {$\scriptstyle{C_k}$};

\filldraw[color=black!60, fill=black!5, thick](0.1+5, 2) circle (0.02);
\filldraw[color=black!60, fill=black!5, thick](0.2+5, 2) circle (0.02);
\filldraw[color=black!60, fill=black!5, thick](0.3+5, 2) circle (0.02);

\draw[black, very thick] (-0.4-5.25 +12,0.8) rectangle (0.4+0.25-5 +12,3.4);

\node[good]      at (0 +7,2.75)      {$\scriptstyle{d+1}$};
\node[good]      at (0 +7,1.5)      {$\scriptstyle{d}$};

\node at (0 +7, 0.5) {$\scriptstyle{C_d}$};

\end{tikzpicture}

%% file: 4n-3.tex
\section{1-out-of-\texorpdfstring{$\mathbf{4\lceil n/3 \rceil}$}{(4n/3)} MMS Algorithm}
\label{sec:4n-3}

\input{technical}
\subsection{Algorithm}
Our algorithm (shown in Algorithm~\ref{algo}) consists of \emph{initialization} and \emph{bag-filling}. First, we remark that assuming $|M| \geq 2n$ is without loss of generality. This is because we can always add dummy goods to $M$ with a value of $0$ for all the agents. The resulting instance is ordered and $d$-normalized if the original instance has these properties.

As mentioned in Section~\ref{sec:tec}, the algorithm first initialize $n$ bags as in \eqref{eq:B_i} (see Figure~\ref{fig:bags}).
\begin{figure}[t]
\centering
\input{figs/B-bags.tikz}
\Description{$n$ bags, where the $i\Th$ bag contains goods $\{i, 2n+1-i\}$.}
\caption{Bag initialization}
\label{fig:bags}
\end{figure}
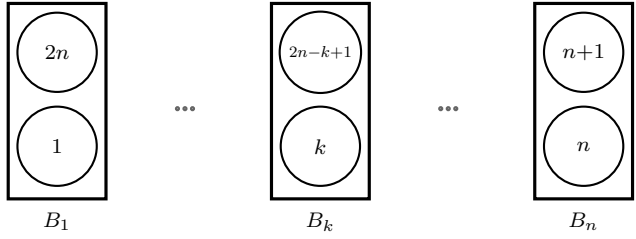
Then, in each round $j$ of bag-filling, we keep adding goods in decreasing value to the bag $B_j$ until some agent with no assigned bag values it at least $1$. Then, we allocate it to an arbitrary such agent.

In the rest of this section, we prove the following theorem, showing the correctness of the algorithm.

\begin{theorem}\label{thm:main}
    Given any ordered $4\ceil{n/3}$-normalized instance, Algorithm \ref{algo} returns a $1$-out-of-$4\ceil{n/3}$ MMS allocation.
\end{theorem}

To do so, it suffices to prove that we never run out of goods in bag-filling. Towards contradiction, assume that the algorithm stops before all agents receive a bundle. Let $i$ be an agent with no bundle. Let $\hat{B}_j$ be the $j$-th bundle after bag-filling.
\begin{observation}\label{upper-bound-left}
    For all $j,k$ such that $j \leq k \leq n$, $v_i(\hat{B}_j) \leq 1 + v_i(2n-k+1)$.
\end{observation}
\begin{proof}
    Let $g$ be the good with the largest index in $\hat{B}_j$. If $g = 2n-j+1$, $v_i(\hat{B}_j \setminus \{g\}) = v_i(j)  \leq 1$ by Proposition \ref{prop:exact}(\ref{exact:1}). If $g > 2n-j+1$, meaning that $g$ was added to $\hat{B}_j$ during bag-filling, then $v_i(\hat{B}_j \setminus \{g\}) < 1$. Otherwise, $g$ would not be added to $\hat{B}_j$. Therefore,
    \begin{align*}
        v_i(\hat{B_j}) &= v_i(\hat{B}_j \setminus \{g\}) + v_i(g) \\
        &\leq 1 + v_i(2n-k+1). &\tag{$v_i(\hat{B}_j \setminus \{g\}) \leq 1$ and $g \geq 2n-k+1$}
    \end{align*}
\end{proof}
\begin{observation}\label{upper-bound-right}
    For all $j,k$ such that $k \leq j \leq n$, $v_i(\hat{B}_j) \leq \max(1 + v_i(2n-k+1), 2v_i(k))$.
\end{observation}
\begin{proof}
    First, assume $\hat{B}_j \neq B_j$ and $g$ be the last good added to $\hat{B}_j$. We have $v_i(\hat{B}_j \setminus \{g\}) < 1$. Otherwise, $g$ would not be added to $\hat{B}_j$. Therefore,
    \begin{align*}
        v_i(\hat{B_j}) &= v_i(\hat{B}_j \setminus \{g\}) + v_i(g) \\
        &< 1 + v_i(2n-k+1). &\tag{$v_i(\hat{B}_j \setminus \{g\}) < 1$ and $g > 2n-k+1$}
    \end{align*}
    Now assume $\hat{B}_j = B_j$. We have
    \begin{align*}
        v_i(\hat{B}_j) &= v_i(B_j) \\
        &= v_i(j) + v_i(2n-j+1) \\
        &\leq 2v_i(k). &\tag{$2n-j+1 > j \geq k$}
    \end{align*}
    Hence, $v_i(\hat{B}_j) \leq \max(1 + v_i(2n-k+1), 2v_i(k))$.
\end{proof}
\begin{observation}\label{less-than-one}
    There exists a bag $B_j$, such that $v_i(B_j) < 1$.
\end{observation}
\begin{proof}
    Otherwise, the algorithm would allocate the remaining bag with the smallest index to agent $i$.
\end{proof}

Let $\ell^*$ be the smallest such that $v_i(B_{\ell^*+1}) < 1$. I.e., $B_{\ell^*+1}$ is the leftmost bag in Figure \ref{fig:bags} with a value less than $1$ to agent $i$. In Section \ref{negative}, we reach a contradiction assuming $v_i(2n-\ell^*)<1/3$ and prove Theorem \ref{contradict-1}.
\begin{restatable}{theorem}{contradictOne}\label{contradict-1}
    If Algorithm \ref{algo} does not allocate a bag to some agent $i$, then $v_i(2n-\ell^*) \geq 1/3$ where $\ell^*$ is the smallest index such that $v_i(B_{\ell^*+1}) < 1$.
\end{restatable}

In Section \ref{positive}, we reach a contradiction assuming $v_i(2n-\ell^*) \geq 1/3$ and prove Theorem \ref{contradict-2}.
\begin{restatable}{theorem}{contradictTwo}\label{contradict-2}
    If Algorithm \ref{algo} does not allocate a bag to some agent $i$, then $v_i(2n-\ell^*) < 1/3$ where $\ell^*$ is smallest such that $v_i(B_{\ell^*+1}) < 1$.
\end{restatable}

By Theorems \ref{contradict-1} and \ref{contradict-2}, agent $i$ who receives no bundle by the end of Algorithm \ref{algo} does not exist, and Theorem \ref{thm:main} follows.

\subsection{\texorpdfstring{$\mathbf{v_i(2n-\ell^*) \geq 1/3}$}{vi(2n-l*) >= 1/3}}
\label{positive}

In this section we assume $v_i(2n-\ell^*) = 1/3 + x$ for $x \geq 0$. 
 We define $A^+ := \{B_1, B_2, \ldots, B_{\ell^*}\}$; see Figure~\ref{k-picture}.
\begin{figure}[t]
    \centering
    \scalebox{\figScale}{\input{figs/k-picture.tikz}}
    \captionAndDescr{An illustration of which group each bag belongs to.}
    \label{k-picture}
\end{figure}
\begin{observation}\label{no-change}
    For all $B_j \in A^+$, $\hat{B}_j = B_j$.
\end{observation}
\begin{proof}
    For all $B_j \in A^+$, $v_i(B_j) \geq 1$. Since $i$ did not receive any bundle, $B_j$ must have been assigned to some other agent, and no good needed to be added to $B_j$ in bag-filling since there is an agent (namely $i$) with no bag who values $B_j$ at least $1$.
\end{proof}
\begin{observation}\label{half-bound}
    For all $j\geq 2n-\ell^*$, $v_i(j) < 1/2$.
\end{observation}
\begin{proof}
    Since $v_i(B_{\ell^*+1}) = v_i(\ell^*+1)+ v_i(2n-\ell^*) < 1$ and $v_i(2n-\ell^*) \leq v_i(\ell^*+1)$, $v_i(2n-\ell^*) < 1/2$.
    Also for all $j \geq 2n-\ell^*$, $v_i(j) \leq v_i(2n-\ell^*) < 1/2$.
\end{proof}

\begin{corollary}[of Observation \ref{half-bound}]
    $x < 1/6$.
\end{corollary}

Let $s$ be the smallest such that either the algorithm stops at step $s+1$ or $B_{s+1}$ gets more than one good in bag-filling.
\begin{observation}
    $s \geq \ell^*$.
\end{observation}
\begin{proof}
    For all $j < \ell^*$, $v_i(B_{j+1}) \geq 1$. Since $i$ did not receive any bundle, $B_{j+1}$, must have been assigned to another agent. Therefore, the algorithm does not stop at step $j+1$. Also, by Observation \ref{no-change}, $B_{j+1}$ gets no good in bag-filling.
\end{proof}
Let $A^1$ be the set of bags in $\{B_{\ell^*+1}, \ldots, B_s\}$ which receive exactly one good in bag-filling. Formally, $A^1 = \{B_j | \ell^* < j \leq s \text{ and } |\hat{B}_j| = 3\}$. Let $A^2 = \{B_1, B_2, \ldots, B_n\} \setminus (A^+ \cup A^1)$.

\begin{lemma}\label{save-2}
    For all $B_j \in A^2$, $v_i(B_j) < 4/3 -2x$.
\end{lemma}
\begin{proof}
    We have
    \begin{align*}
        1 &> v_i(B_{\ell^*+1}) \\
        &= v_i(\ell^*+1) + v_i(2n-\ell^*) \\
        &= v_i(\ell^*+1) + \frac{1}{3} + x.
    \end{align*}
    Hence, $v_i(\ell^*+1) < 2/3 - x$. Also, for $B_j \in A^2$, we have
    \begin{align*}
        v_i(B_j) &= v_i(j) + v_i(2n-j+1) \\
        &\leq 2v_i(\ell^*+1) &\tag{$2n-j+1 > j \geq \ell^*+1$ since $B_j \in A^2$} \\
        &< \frac{4}{3} - 2x.
    \end{align*}
    So if $\hat{B}_j = B_j$, the inequality holds. Now assume $\hat{B}_j \neq B_j$. This implies that $j \geq s+1$ and the algorithm did not stop at step $j$ before adding a good to $B_j$. Therefore it did not stop at step $s +1$ before adding a good to $B_{s+1}$ either. Let $g$ be the first good added to $B_{s+1}$. Since $B_{s+1}$ requires more than one good,
    \begin{align*}
        1 &> v_i(B_{s +1} \cup \{g\}) \\
        &= v_i(s+1) + v_i(n-s) + v_i(g) \\
        &\geq 2v_i(2n-\ell^*) + v_i(g) &\tag{$s+1 < n-s \leq 2n-\ell^*$} \\
        &= \frac{2}{3} + 2x + v_i(g).
    \end{align*}
    Therefore, $v_i(g) < 1/3 - 2x$. Now let $h$ be the last good added to bag $B_j$. We have
    \begin{align*}
        v_i(\hat{B}_j) &= v_i(\hat{B}_j \setminus \{h\}) + v_i(h) \\
        &< 1 + v_i(g) &\tag{$v_i(\hat{B}_j \setminus \{h\}) < 1$ and $v_i(h) \leq v_i(g)$} \\
        &< \frac{4}{3} - 2x.
    \end{align*}
\end{proof}

\begin{lemma}\label{save-all}
    For all $B_j \in A^+ \cup A^1$, $v_i(\hat{B}_j) \leq 4/3 + x$.
\end{lemma}
\begin{proof}
    First assume $B_j \in A^+$. We have $j \leq \ell^*$. Also,
    \begin{align*}
        v_i(\hat{B}_j) &= v_i(B_j) &\tag{$\hat{B}_j = B_j$} \\
        &= v_i(j) + v_i(2n-j+1) \\
        &\leq 1 + v_i(2n-\ell^*) &\tag{$v_i(j) \leq 1$ and $2n-j+1 > 2n-\ell^*$} \\
        &= \frac{4}{3} +x.
    \end{align*}
    Now assume $B_j \in A^1$. Let $g$ be the good added to bag $B_j$ in bag-filling. We have,
    \begin{align*}
        v_i(\hat{B}_j) &= v_i(B_j) + v_i(g) \\
        &< 1 + v_i(2n-\ell^*) &\tag{$v_i(B_j) < 1$ and $v_i(g) \leq v_i(2n+1) \leq v_i(2n-\ell^*)$}\\
        &= \frac{4}{3} + x.
    \end{align*}
\end{proof}

Let $|A^1|=2n/3+\ell$. Then $|A^2| = n - \ell^* - (2n/3 + \ell) = n/3-(\ell+\ell^*)$. If $\ell + \ell^* \leq 0$, then $|A^2| \geq n/3$ and hence there are at least $n/3$ bags with value less than $4/3 -2x$ (by Lemma \ref{save-2}) and at most $2n/3$ bags with value at most $4/3 + x$ (by Lemma \ref{save-all}). Hence,
\[ v_i(M) < \frac{n}{3}(\frac{4}{3}-2x) + \frac{2n}{3}(\frac{4}{3}+x) = \frac{4n}{3}, \]
which is a contradiction since $v_i(M) = 4n/3$. So assume $\ell + \ell^* > 0$.

Limit the goods in a $1$-out-of-$4n/3$ MMS partition $P^i=(P^i_1, \dots, P^i_{4n/3})$ of agent $i$ to $\{1, \ldots, 8n/3+\ell\}$ and let $Q$ be the set of bags in $P^i$ containing goods $\{1, \ldots, \ell^*\}$. Formally, $Q = \{ P_j^i \cap \{1, \ldots, 8n/3+\ell\}: |P_j^i \cap \{1, \ldots, \ell^*\}| \geq 1\}$. Let $t$ be the number of bags of size $1$ in $Q$.

\begin{figure}[!hbt]
    \centering
    \scalebox{\figScale}{\input{figs/bag-filling.tikz}}
    \captionAndDescr{The first $2\ell + t + 2\ell^* -2n/3$ goods are marked with red, and the goods considered in Lemma \ref{expowerful} are marked with blue.}
    \label{colorful}
\end{figure}

\begin{lemma}\label{expowerful}
    Let $t$ be the number of bags of size $1$ in $Q = \{ P_j^i \cap \{1, \ldots 8n/3+\ell\}: |P_j^i \cap \{1, \ldots \ell^*\}| \geq 1\}$. Then,
    \begin{align*}
        v_i(&\{8n/3 - 2\ell-t-2\ell^*+1, \ldots, 8n/3 + \ell\} \\
        &\cup \{t+1, \ldots, \ell^*\} \\
        &\cup \{2n-\ell^*+1, \ldots, 2n-t\}) \leq 2\ell^* + \ell -t.
    \end{align*}
\end{lemma}
The goods considered in Lemma \ref{expowerful} are marked with blue in Figure \ref{colorful}.
First, we prove that the goods mentioned in Lemma \ref{expowerful} are distinct. To that end, it suffices to prove that $8n/3 - 2\ell-t-2\ell^*+1 > 2n-t$. It follows from the fact that $2n/3+\ell+\ell^* \leq n$.
Before proving Lemma \ref{expowerful}, let us show how to obtain a contradiction assuming this lemma holds.
Note that since there are $n/3 - \ell - \ell^*$ bags with value less than $4/3-2x$ (namely the bags in $A^2$), it suffices to prove that there exists $3(\ell + \ell^*)$ other bags with total value $4(\ell + \ell^*)$. Since the remaining $2n/3 -2\ell -2\ell^*$ bags are of value at most $4/3+x$ (by Lemma \ref{save-all}), we get
\begin{align}
    v_i(M) < (\frac{n}{3} - \ell - \ell^*) (\frac{4}{3} - 2x) + (\frac{2n}{3} - 2\ell - 2\ell^*)(\frac{4}{3} + x) + 4 (\ell + \ell^*) = \frac{4n}{3} \label{important-ineq}
\end{align}
which is a contradiction since $v_i(M)=4n/3$.

Now consider $B = \{\hat{B}_1, \ldots, \hat{B}_{2\ell + t + 2\ell^* -2n/3}, \hat{B}_{t+1}, \ldots, \hat{B}_{\ell^*} \} \cup \hat{A}^1$ where $\hat{A}^1$ is the set of bags in $A^1$ after bag-filling. $B$ consists of $3(\ell+\ell^*)$ bags. Now we prove that $v_i(\bigcup_{B_j \in B}B_j) \leq 4(\ell+\ell^*)$. We have
\begin{align*}
    v_i(\bigcup_{\hat{B}_j \in B}\hat{B}_j) \leq v_i(&\bigcup_{B_j \in A^1} B_j) \\
    + v_i(\textcolor{red}{\{}&\textcolor{red}{1, \ldots, 2\ell + t + 2\ell^* -2n/3\}})\\
    + v_i(\textcolor{blue}{\{}&\textcolor{blue}{8n/3 - 2\ell-t-2\ell^*+1, \ldots, 8n/3 + \ell\}} \\
     &\textcolor{blue}{\cup \{t+1, \ldots, \ell^*\}} \\
     &\textcolor{blue}{\cup \{2n-\ell^*+1, \ldots, 2n-t\}}) .
\end{align*}
We bound the value of the goods marked with different colors in different inequalities.
\begin{observation}
    For all $B_j \in A^1$, $v_i(B_j) < 1$.
\end{observation}
Since $|A^1| = 2n/3 + \ell$,
\begin{align*}
    v_i(\bigcup_{B_j \in A^1} B_j) < 2n/3 +\ell.
\end{align*}
Also, since all goods are of value at most $1$ to agent $i$,

\begin{align*}
    \textcolor{red}{v_i(\{1, \ldots, 2\ell + t + 2\ell^* -2n/3\}) \leq 2\ell + t + 2\ell^* -2n/3}.
\end{align*}
By Lemma \ref{expowerful},
\textcolor{blue}{
\begin{align*}
        v_i(&\{8n/3 - 2\ell-t-2\ell^*+1, \ldots, 8n/3 + \ell\} \\
        &\cup \{t+1, \ldots, \ell^*\} \\
        &\cup \{2n-\ell^*+1, \ldots, 2n-t\}) \leq 2\ell^* + \ell -t.
\end{align*}}

By adding all the inequalities, we get
\begin{align*}
    v_i(\bigcup_{\hat{B}_j \in B}\hat{B}_j) \leq 4(\ell+\ell^*).
\end{align*}
Hence, Inequality \eqref{important-ineq} holds, which is a contradiction. So the case of $v_i(2n-\ell^*) \geq 1/3$ cannot arise.
\contradictTwo*

In the rest of this section, we prove Lemma \ref{expowerful}.

\subsubsection{Proof of Lemma \ref{expowerful}}

To prove Lemma \ref{expowerful}, we partition the goods considered in this lemma into two parts. These parts are colored red and blue in Figure \ref{colorful-1}. We bound the value of red goods in Lemma~\ref{k-t}, i.e.,
\[ \sum_{2n-\ell^* < j\leq 2n-t}v_i(j) + \sum_{8n/3 - 2\ell-t-2\ell^* < j \leq 8n/3 - 2\ell - 2t - \ell^*}v_i(j) < \ell^*-t, \]
and the value of the blue goods in Lemma~\ref{k+l}, i.e.,
\[ \sum_{t < j \leq \ell^*} v_i(j) + \sum_{8n/3 - 2\ell - 2t - \ell^* < j \leq 8n/3 + \ell}v_i(j) \leq \ell^* + \ell. \]
Thereafter, we have
\begin{align*}
    v_i(&\{8n/3 - 2\ell-t-2\ell^*+1, \ldots, 8n/3 + \ell\} \\
    &\cup \{t+1, \ldots, \ell^*\} \\
    &\cup \{2n-\ell^*+1, \ldots, 2n-t\}) \\
    &= \textcolor{red}{\sum_{2n-\ell^* < j\leq 2n-t}v_i(j) + \sum_{8n/3 - 2\ell-t-2\ell^* < j \leq 8n/3 - 2\ell - 2t - \ell^*}v_i(j)} \\
    &+ \textcolor{blue}{\sum_{t < j \leq \ell^*} v_i(j) + \sum_{8n/3 - 2\ell - 2t - \ell^* < j \leq 8n/3 + \ell}v_i(j)} \\
    &< (\ell^*-t) + (\ell^*+\ell) &\tag{Lemma \ref{k-t} and \ref{k+l}}\\
    &= 2\ell^*+\ell-t,
\end{align*}
and Lemma \ref{expowerful} follows.

It suffices to prove Lemmas \ref{k-t} and \ref{k+l}. In the rest of this section, we prove these two lemmas.
\begin{figure}[!hbt]
    \centering
    \scalebox{\figScale}{\input{figs/Lemma-Fig.tikz}}
    \captionAndDescr{The goods considered in in Lemma~\ref{k-t} are marked with red and the goods in Lemma~\ref{k+l} are marked with blue.}
    \label{colorful-1}
\end{figure}
Limit the goods in a $1$-out-of-$4n/3$ MMS partition of agent $i$ to $\{1, \ldots, 8n/3+\ell\}$ and let $R$ be the set of the resulting bags. Formally, for all $j \in [4n/3]$, $R_j = P_j^{i} \cap \{1, \ldots, 8n/3+\ell\}$ and $R = \{R_1, \ldots, R_{4n/3}\}$. Without loss of generality, assume $|R_1| \geq |R_2| \geq \ldots \geq |R_{4n/3}|$. Let $t$ be the number of bags of size $1$ in $R$.

\begin{lemma}\label{powerful}
    If there exist $t$ bags of size at most $1$ in $R$, then
    \[ \sum_{1 \leq j \leq t+\ell}|R_j| \geq 3(t+\ell). \]
\end{lemma}
\begin{proof}
    Since $R_j$'s are sorted in decreasing order of their size,
    \begin{align*}
        \sum_{1 \leq j \leq t+\ell}|R_j| \geq (t+\ell)|R_{t+\ell}|.
    \end{align*}
    Hence, if $|R_{t+\ell}| \geq 3$, then $\sum_{1 \leq j \leq t+\ell}|R_j| \geq 3(t+\ell).$ So assume $|R_{t+\ell}| \leq 2$.
    \begin{align*}
        \frac{8n}{3} + \ell &= \sum_{1 \leq j \leq 4n/3}|R_j| \\
        &= \sum_{1 \leq j \leq t+\ell} |R_j| + \sum_{t+\ell < j \leq 4n/3 - t}|R_j| + \sum_{4n/3-t < j \leq 4n/3}|R_j| \\
        &\leq \sum_{1 \leq j \leq t+\ell} |R_j| +  (\frac{4n}{3} - 2t - \ell)|R_{t+\ell}| + t \\
        &\leq \sum_{1 \leq j \leq t+\ell} |R_j| + 2(\frac{4n}{3} - 2t - \ell) + t
    \end{align*}
    Therefore, \[ \sum_{j \in [t+\ell]} |R_j| \geq 3(\ell+t). \]
\end{proof}
\begin{lemma}\label{bound}
    $\ell+\ell^*+t \leq 4n/3$.
\end{lemma}
\begin{proof}
    We have $\ell^* + 2n/3 + \ell \leq s \leq n$. See Figure \ref{colorful} for intuition. Therefore, $\ell^* + \ell \leq n/3$. Also, $t \leq \ell^* \leq n$. Hence $\ell+\ell^*+t \leq 4n/3$.
\end{proof}
\begin{lemma}\label{k-t}
    \textcolor{red}{
    \[ \sum_{2n-\ell^* < j\leq 2n-t}v_i(j) + \sum_{8n/3 - 2\ell-t-2\ell^* < j \leq 8n/3 - 2\ell - 2t - \ell^*}v_i(j) < \ell^*-t. \]
    }
\end{lemma}
\begin{proof}
    Let $B' = \{2n-\ell^*+1, \ldots, 2n-t\} \cup \{8n/3 - 2\ell-t-2\ell^*+1, \ldots, 8n/3 - 2\ell - 2t - \ell^*\}$. $|B'| = 2(\ell^*-t)$ and by Observation \ref{half-bound} for all goods $g \in B'$, $v_i(g) < 1/2$. Therefore, $v_i(B') < \ell^*-t$.
\end{proof}
\begin{lemma}\label{k+l}
    \textcolor{blue}{
\[ \sum_{t < j \leq \ell^*} v_i(j) + \sum_{8n/3 - 2\ell - 2t - \ell^* < j \leq 8n/3 + \ell}v_i(j) \leq \ell^* + \ell. \]
    }
\end{lemma}
\begin{proof}
    Recall that $\{R_1, \ldots, R_{4n/3}\}$ is the set of bags in the $1$-out-of-$4n/3$ MMS partition of agent $i$ after removing goods $\{8n/3+\ell+1, \ldots, m\}$. Moreover, we know exactly $t$ of these bags have size $1$. If there is a bag $R_j = \{g\}$ for $g>t$, there must be a good $g' \in [t]$ such that $g' \in R_{j'}$ and $|R_{j'}|>1$. Swap the goods $g$ and $g'$ between $R_j$ and $R_{j'}$ as long as such good $g$ exists. Note that $v_i(R_{j'})$ can only decrease and $v_i(R_j)=v_i(g') \leq 1$. Therefore, in the end of this process for all $j \in [4n/3]$, $v_i(R_j) \leq 1$ and we can assume bags containing goods $1, \ldots, t$ are of size $1$ and bags containing goods $t+1, \ldots, \ell^*$ 
    are of a size more than $1$. Recall that $|R_1| \geq \ldots \geq |R_{4n/3}|$. Let $T_{j}$ be the bag that contains good $j$.

    Consider the bags $B=\{R_1, \ldots, R_{t+\ell}\} \cup \{T_{t+1}, \ldots, T_{\ell^*}\}$. If $|B| < \ell^*+\ell$, keep adding a bag with the largest number of goods to $B$ until there are exactly $\ell^*+\ell$ bags in $B$. First we show that $B$ contains at least $3\ell + 2\ell^* + t$ goods. Namely,
    \[ \sum_{S \in B} |S| \geq 3\ell + 2\ell^* + t. \]
    By Lemma \ref{powerful}, $\sum_{1 \leq j \leq t+\ell} |R_j| \geq 3(t+\ell)$. If all the remaining $\ell^*-t$ bags in $B \setminus \{R_1, \ldots, R_{t+\ell}\}$ are of size $2$, then $\sum_{S \in B} |S| \geq 3(t+\ell)+2(\ell^*-t) = 3\ell + 2\ell^* +t$. Otherwise, there is a bag in $B$ of size at most $1$; hence, all bags outside $B$ are also of size at most $1$. So we have
    \begin{align*}
        \frac{8n}{3} + \ell &= \sum_{S \in B} |S| + \sum_{S \notin B} |S| \\
        &\leq \sum_{S \in B} |S| + (\frac{4n}{3} - \ell^* - \ell).
    \end{align*}
    Therefore,
    \begin{align*}
        \sum_{S \in B} |S| &\geq 4n/3 + 2\ell + \ell^* \\
        &\geq 3\ell + 2\ell^* +t. &\tag{Lemma \ref{bound}}
    \end{align*}
    Note that the goods $\{t+1, \ldots, \ell^*\}$ are contained in $B$ and moreover, $B$ contains at least $3\ell + 2\ell^* +t - (\ell^*-t) = 3\ell + 2t + \ell^*$ other goods. Therefore,
    \begin{align*}
        \ell^*+\ell &\geq \sum_{S \in B} v_i(S) \\
        &\geq \sum_{t < j \leq \ell^*} v_i(j) + \sum_{8n/3 - 2\ell - 2t - \ell^* < j \leq 8n/3 + \ell} v_i(j).
    \end{align*}
The last inequality follows because we used the $3\ell + 2t + \ell^*$ lowest valued goods in $[8n/3+\ell]$.
\end{proof}

\subsection{\texorpdfstring{$\mathbf{v_i(2n-\ell^*)<1/3}$}{vi(2n-l*) < 1/3}}
\label{negative}

Let $r^*$ be largest such that $v_i(B_{r^*})<1$. That is, $B_{r*}$ is the rightmost bag in Figure \ref{fig:bags} with a value less than $1$ to agent $i$.
\begin{lemma}\label{r-star-bound}
    If $v_i(2n-r^*+1) \leq 1/3$, then $r^* < 2n/3$.
\end{lemma}
\begin{proof}
    Since $1>v_i(B_{r^*})=v_i(r^*)+v_i(2n-r^*+1)$, we have $v_i(r^*)<2/3+x$.
    By Observation \ref{upper-bound-left}, for all $j \leq r^*$, $v_i(\hat{B}_j) \leq 4/3-x$.
    Also, by Observation \ref{upper-bound-right}, for all $j>r^*$, $v_i(\hat{B}_j)<4/3+2x$.
    Hence, we have
    \begin{align*}
        \frac{4n}{3} &= v_i(M) \\
        &= \sum_{j \leq r^*} v_i(\hat{B}_j) + \sum_{j > r^*} v_i(\hat{B}_j) \\
        &< r^*(\frac{4}{3}-x) + (n-r^*)(\frac{4}{3}+2x) \\
        &= \frac{4n}{3} + x(2n-3r^*).
    \end{align*}
    Therefore, $r^* < 2n/3$.
\end{proof}
\begin{lemma}
    $v_i(2n-r^*+1) > 1/3$.
\end{lemma}
\begin{proof}
    Towards contradiction, assume $v_i(2n-r^*+1) = 1/3-x$ for $x \geq 0$. By Lemma \ref{r-star-bound}, $r^* < 2n/3$.
    \begin{claim}\label{claim-r-star}
        $\sum_{j > r^*} v_i(\hat{B}_j) < \frac{10n}{9}-r^* + \frac{2nx}{3}$.
    \end{claim}
    \begin{claimproof}
    Note that by the definition of $r^*$, for all $j > r^*$, $\hat{B}_j = B_j$.
    By Lemma \ref{lem:Csum}, $v_i(\{2n/3+r^*+1, \ldots, 2n-r^*\}) \leq 2n/3-r^*$. Also since $v_i(r^*) < 2/3+x$, $v_i(\{r^*+1, \ldots, 2n/3+r^*\}) \leq \frac{2n}{3}(\frac{2}{3}+x)$. In total, we get
    \begin{align*}
        \sum_{j > r^*} v_i(\hat{B}_j) &= \sum_{j > r^*} v_i(B_j) \\
        &= v_i(\{r^*+1, \ldots, 2n/3+r^*\}) + v_i(\{2n/3 +r^*+1, \ldots, 2n-r^*\}) \\
        &< \frac{2n}{3}(\frac{2}{3}+x) + \frac{2n}{3}-r^* \\
        &= \frac{10n}{9}-r^* + \frac{2nx}{3}.
    \end{align*}
    Therefore, Claim \ref{claim-r-star} holds.
    \end{claimproof}
    We have
    \begin{align*}
        \frac{4n}{3} &= v_i(M) \\
        &= \sum_{j \leq r^*} v_i(\hat{B}_j) + \sum_{j > r^*} v_i(\hat{B}_j) \\
        &< r^*(\frac{4}{3} -x) + \frac{10n}{9}-r^* + \frac{2nx}{3} &\tag{Observation \ref{upper-bound-left} and Claim \ref{claim-r-star}}\\
        &= r^*(\frac{1}{3}-x) + \frac{10n}{9} + \frac{2nx}{3}.
    \end{align*}
    Thus,
    \begin{align*}
        \frac{2n}{9} &< r^*(\frac{1}{3}-x) + \frac{2n}{3}(x) \\
        &\leq \frac{2n}{3} \cdot \frac{1}{3}, \tag{$r^* \leq 2n/3$ by Lemma \ref{r-star-bound}}
    \end{align*}
    which is a contradiction. Hence, $v_i(2n-r^*+1) > 1/3$.
\end{proof}

Recall that $\ell^*$ be the smallest such that $v_i(B_{\ell^*+1}) < 1$, i.e., $B_{\ell^*+1}$ is the leftmost bag in Figure \ref{fig:bags} with value less than $1$ to agent $i$. Let $\ell$ be largest such that $v_i(B_\ell)<1$ and $v_i(2n- \ell +1) \leq 1/3$. Since $v_i(B_{\ell^*+1})<1$ and $v_i(2n- \ell^*) < 1/3$, such an index exists and $\ell \ge \ell^*+1$. Also, let $r$ be smallest such that $v_i(B_{r+1})<1$ and $v_i(2n-r) \geq 1/3$. Again, since $v_i(B_{r^*})<1$ and $v_i(2n- r^*+1) > 1/3$, such an index exists. We set $x := 1/3 - v_i(2n- \ell +1)$ and $y := v_i(2n-r) - 1/3$. See Figure \ref{r-l-fig}.
\begin{observation}\label{lessThan13}
    $x < 1/3$.
\end{observation}
\begin{proof}
    Towards a contradiction, assume $x=1/3$. Therefore, $v_i(2n-\ell+1)=0$. Let $k < 2n-\ell+1$ be the number of goods with a value larger than $0$ to agent $i$. Consider $(P^i_1 \cap [k], \ldots, P^i_{4n/3} \cap [k])$. There are at least $\ell$ many indices $j$ such that, $|P^i_j \cap [k]|=1$. Since $\mathcal{I}$ is $4n/3$-normalized, $v_i(1) = \ldots = v_i(\ell)=1$ which is a contradiction with $v_i(B_{\ell^*+1})<1$.
\end{proof}

\begin{observation}\label{y-bound}
    $y< 1/6$.
\end{observation}
\begin{proof}
    We have $1/3 + y = v_i(2n-r) \leq v_i(B_r)/2 < 1/2$. Thus, $y<1/6$.
\end{proof}
\begin{figure}
    \centering
    \scalebox{\figScale}{\input{figs/r-l-fig.tikz}}
    \captionAndDescr{$v_i(2n-\ell+1)=1/3-x$ and $v_i(2n-r)=1/3+y$.}
    \label{r-l-fig}
\end{figure}
\begin{corollary}[of Observation \ref{upper-bound-left}]\label{cor-upper-left}
    For all $j \leq \ell$, $v_i(\hat{B}_j) \leq 4/3-x$.
\end{corollary}
\begin{corollary}[of Observation \ref{upper-bound-right}]\label{cor-upper-right}
    For all $j>r$, $v_i(\hat{B}_j) \leq \max(4/3-x, 4/3-2y)$.
\end{corollary}
\begin{observation}\label{middle-block}
    For all $\ell < j \leq r$, $1 \leq v_i(\hat{B}_j) < 1+x+y$.
\end{observation}
\begin{proof}
    Note that by definition of $\ell$ and $r$, for all $\ell < j \leq r$, $v_i(B_j) \geq 1$. Therefore, $\hat{B}_j=B_j$. Also,
    \begin{align*}
        v_i(B_j) &= v_i(j) + v_i(2n-j+1) \\
        &\leq v_i(\ell) + v_i(2n-r) \tag{$\ell < j$ and $2n-r < 2n-j+1$}\\
        &< (\frac{2}{3}+x) + (\frac{1}{3}+y) \tag{$v_i(B_\ell)<1$ and $v_i(B_{r+1})<1$}\\
        &= 1+x+y.
    \end{align*}
\end{proof}
\begin{lemma}
    $r-\ell > 2n/3$.
\end{lemma}
\begin{proof}
    If $x+y \leq 1/3$, then by Corollaries \ref{cor-upper-left} and \ref{cor-upper-right} and Observation \ref{middle-block}, for all $t \in [n]$ we have $v_i(\hat{B}_t) \leq 4/3$ and for at least one bag this value is less than $1$ by Observation \ref{less-than-one}. Therefore, $v_i(M) < 4n/3$, which is a contradiction. Thus, $x+y>1/3$.
    We have
    \begin{align*}
        \frac{4n}{3} &= v_i(M) \\
        &= \sum_{j \leq \ell} v_i(\hat{B}_j) + \sum_{\ell < j \leq r} v_i(\hat{B}_j) + \sum_{j > r} v_i(\hat{B}_j) \\
        &\leq \ell(\frac{4}{3}-x) + (r-\ell)(1+x+y) + (n-r) \max(\frac{4}{3}-x, \frac{4}{3}-2y) \\
        &\leq (r-\ell)(1+x+y) + (n-r+\ell) \max(\frac{4}{3}-x, \frac{4}{3}-2y) \tag{Corollaries \ref{cor-upper-left} and \ref{cor-upper-right} and Observation \ref{middle-block}}\\
        &= \frac{4n}{3} + (r-\ell)(x+y-\frac{1}{3}) - (n-r+\ell)\min(x,2y).
    \end{align*}
    Therefore, $(r-\ell)(x+y-1/3) \geq (n-r+\ell)\min(x,2y)$. By Observation \ref{lessThan13}, $x<1/3$ and thus, we have $x+y-1/3 < y$. Also, since $y < 1/6$ (by Observation \ref{y-bound}), we have $x+y-1/3 < x-1/6 < x/2$. Thus, $x+y-1/3 < \min(x,2y)/2$. Hence, $r-\ell > 2(n-r+\ell)$ and therefore, $r-\ell > 2n/3$.
\end{proof}
Let $r - \ell = 2n/3 + s$. Recall that $P^{i} = (P^{i}_1, \ldots, P^{i}_{4n/3})$ is an $(4n/3)$-MMS partition of $M$ for agent $i$. Since $i$ is fixed, we use $P = (P_1, \ldots, P_{4n/3})$ instead for ease of notation. For all $j \in [4n/3]$, let $g_j$ be good with the smallest index (and hence the largest value) in $P_j$. Without loss of generality, assume $g_1 < g_2 < \ldots < g_{4n/3}$. Observe that $\{1,\dots,r\} \subseteq \cup_{k\in [r]} P_k$. Let $S'$ be the set of goods in $\{r+1, \ldots, 2n-\ell\}$ that appear in the first $r$ bags in $P$. Formally, $S' = \{g \in \{r+1, \ldots, 2n-\ell\} \mid g \in \cup_{j \in [r]} P_j\}$. Let $s':=\min(|S'|,s)$.
\begin{figure}[!hbt]
    \centering
    \scalebox{\figScale}{\input{figs/Lemma-Fig-1.tikz}}
    \captionAndDescr{The goods considered in Lemma \ref{difficult-bound} are marked with blue.}
    \label{colorful-11}
\end{figure}
\begin{restatable}{lemma}{difficultbound} \label{difficult-bound}
    $v_i(\{r-s'+1, \ldots, r\} \cup \{2n-\ell-3s+2s'+1, \ldots, 2n-\ell\}) \leq s$.
\end{restatable}
The goods considered in Lemma \ref{difficult-bound} are marked with blue in Figure \ref{colorful-11}.
Before proving Lemma \ref{difficult-bound}, let us assume it holds and reach a contradiction. Since $v_i(\ell) < 1-v_i(2n-\ell+1)= 2/3+x$, we have
\begin{align}
    v_i(\{\ell+1, \ldots, r-s'\}) &< (\frac{2n}{3}+s-s')(\frac{2}{3}+x). \label{easy-1}
\end{align}
Also, since $v_i(2n-r+1) =1/3+y$,
\begin{align}
    v_i(\{2n-r+1, \ldots, 2n-\ell-3s+2s'\}) &\leq (\frac{2n}{3}-2s+2s')(\frac{1}{3}+y). \label{easy-2}
\end{align}
Therefore,
\begin{align*}
    \sum_{\ell < j \leq r} v_i(\hat{B}_j) &= \sum_{\ell < j \leq r} v_i(B_j) \\
    &= v_i(\{\ell+1, \ldots, r\} \cup \{2n-r+1, \ldots, 2n-\ell\}) \\
    &= v_i(\{\ell+1, \ldots, r-s'\}) \\
    &\indent + v_i(\{r-s'+1, \ldots, r\} \cup \{2n-\ell-3s+2s'+1, \ldots, 2n-\ell\}) \\
    &\indent + v_i(\{2n-r+1, \ldots, 2n-\ell-3s+2s'\}) \\
    &< (\frac{2n}{3}+s-s')(\frac{2}{3}+x) + s + (\frac{2n}{3}-2s+2s')(\frac{1}{3}+y) \tag{Inequalities \eqref{easy-1} and \eqref{easy-2} and Lemma \ref{difficult-bound}}\\
    &= \frac{2n}{3}(1+x+y)+(s-s')(x-2y)+s.
\end{align*}
Thus,
\begin{align*}
    \frac{4n}{3} &= v_i(M) \\
    &= \sum_{j \leq \ell} v_i(\hat{B}_j) + \sum_{\ell < j \leq r} v_i(\hat{B}_j) + \sum_{j>r} v_i(\hat{B}_j) \\
    &< (\ell+n-r)\max(\frac{4}{3}-x, \frac{4}{3}-2y) + \frac{2n}{3}(1+x+y)+(s-s')(x-2y)+s \tag{Corollaries \ref{cor-upper-left} and \ref{cor-upper-right}} \\
    &= (\frac{n}{3}-s)\max(\frac{4}{3}-x, \frac{4}{3}-2y) + \frac{2n}{3}(1+x+y)+(s-s')(x-2y)+s.
\end{align*}
If $x \leq 2y$, then by replacing $\max(4/3-x, 4/3-2y)$ with $4/3-x$ in the above inequality, we get
\begin{align*}
    \frac{4n}{3} &< (\frac{n}{3}-s)(\frac{4}{3}-x) + \frac{2n}{3}(1+x+y)+(s-s')(x-2y)+s \\
    &\leq \frac{n}{3}(\frac{4}{3}-x) + \frac{2n}{3}(1+x+y)+(s-s')(x-2y) \tag{$4/3-x \geq 1$}\\
    &\leq \frac{n}{3}(\frac{10}{3}+x+2y) \tag{$(s-s')(x-2y) \leq 0$} \\
    &< \frac{4n}{3}, \tag{$x \leq 1/3$ and $y<1/6$}
\end{align*}
which is a contradiction. If $2y < x$, by replacing $\max(4/3-x, 4/3-2y)$ with $4/3-2y$, we get
\begin{align*}
    \frac{4n}{3} &< (\frac{n}{3}-s)(\frac{4}{3}-2y) + \frac{2n}{3}(1+x+y)+(s-s')(x-2y)+s \\
    &= \frac{n}{3}(\frac{10}{3}+2x) - s(\frac{1}{3}-x) - s'(x-2y) \\
    &\leq \frac{n}{3}(\frac{10}{3}+2x) &\tag{$x \leq 1/3$ and $x>2y$} \\
    &\leq \frac{4n}{3}, \tag{$x \leq 1/3$}
\end{align*}
which is again a contradiction. Therefore, it is not possible that $v_i(2n-\ell^*) < 1/3$. Thus, Theorem \ref{contradict-1} follows.
\contradictOne*

It only remains to prove Lemma \ref{difficult-bound}. The main idea is as follows. Recall that $s'=\min(|S'|,s)$. We consider two cases for $s'$. If $s' = s$, then in order to prove Lemma \ref{difficult-bound}, we must prove
\[ v_i(\{r-s'+1, \ldots, r\} \cup \{2n-\ell - s'+1, \ldots, 2n-\ell\}) \leq s', \]
which is what we do in Claim \ref{claim-2}.
In case $s' = |S'|$, we prove
\[ v_i(\{r-s'+1, \ldots, r\}) + v_i(S') \leq s' \]
in Claim \ref{claim-22} and
\[ v_i (\{2n-\ell-3s+2s'+1, \ldots, 2n-\ell\}) - v_i(S') \leq s-s' \]
in Claim \ref{claim-3}. Adding the two sides of the inequalities implies Lemma \ref{difficult-bound}. We prove this lemma in Section \ref{proof-sec-1}.

\subsubsection{Proof of Lemma \ref{difficult-bound}}\label{proof-sec-1}
\difficultbound*
    Note that $\{1, \ldots, r\} \cup S' \subseteq P_1 \cup \ldots \cup P_r$.
    For $j \in [r]$, let $Q_j = P_j \cap (\{1, \ldots, r\} \cup S')$. We begin with proving the following claim. 
    \begin{claim}\label{claim-1}
        There are $s'$ many sets like $Q_{j_1}, \ldots, Q_{j_{s'}}$ such that $|\cup_{k \in [s']} Q_{j_k}| \geq 2s'$ and $|\cup_{k \in [s']} Q_{j_k} \cap \{1, \ldots, r\}| \geq s'$.
    \end{claim}
    \begin{proof}
        If $s'=0$, the claim trivially holds. Thus, assume $s' \geq 1$. By induction, we prove that for any $t \leq s'$, there are $t$ many sets like $Q_{j_1}, \ldots, Q_{j_t}$ such that $|\cup_{k \in [t]} Q_{j_k}| \geq 2t$ and $|\cup_{k \in [t]} Q_{j_k} \cap \{1, \ldots, r\}| \geq t$.
        \paragraph{\boldmath Induction basis: $t=1$.} If there exists $Q_k$ such that $|Q_k \cap \{1, \ldots, r\}| \geq 2$, let $j_1 = k$. Otherwise, for all $k \in [r]$, we have $|Q_k \cap \{1, \ldots, r\}| = 1$. Since $s' \geq 1$, there must be an index $k$ such that $|Q_k \cap S'| \geq 1$. Let $j_1=k$.

        \paragraph{\boldmath Induction assumption:} There are $t$ many sets like $Q_{j_1}, \ldots, Q_{j_t}$ such that $|\cup_{k \in [t]} Q_{j_k}| \geq 2t$ and $|\cup_{k \in [t]} Q_{j_k} \cap \{1, \ldots, r\}| \geq t$.

        Now for $t+1 \leq s'$, we prove that there are $t+1$ many sets like $Q_{j_1}, \ldots, Q_{j_{t+1}}$ such that $|\cup_{k \in [t+1]} Q_{j_k}| \geq 2t+2$ and $|\cup_{k \in [t+1]} Q_{j_k} \cap \{1, \ldots, r\}| \geq t+1$.
        \paragraph{\boldmath Case 1: $|\cup_{k \in [t]} Q_{j_k}| \geq 2t+2$:} If $|\cup_{k \in [t]} Q_{j_k} \cap \{1, \ldots, r\}| \geq t+1$, set $j_{t+1}=k$ for an arbitrary $k \in [r] \setminus \{j_1, \ldots, j_t\}$. Otherwise, set $j_{t+1}=k$ for an index $k \in [r] \setminus \{j_1, \ldots, j_t\}$ such that $|Q_{k} \cap \{1, \dots, r\}| \ge 1$.
        \paragraph{\boldmath Case 2: $|\cup_{k \in [t]} Q_{j_k}| = 2t+1$:} If there exists $k \in [r] \setminus \{j_1, \ldots, j_t\}$, such that $|Q_k \cap [r]| \geq 1$, set $j_{t+1}=k$. Otherwise, set $j_{t+1}=k$ for any $k \in [r] \setminus \{j_1, \ldots, j_t\}$ such that $|Q_k| \geq 1$. Since $|\cup_{j \in [r]} Q_j| \geq r+s' > 2t+1$, such $k$ exists.
        \paragraph{\boldmath Case 3. $|\cup_{k \in [t]} Q_{j_k}| = 2t$ and $|\cup_{k \in [t]} Q_{j_k} \cap \{1, \ldots, r\}| \geq t+1$:} $|\cup_{k \in [r] \setminus \{j_1, \ldots, j_t\}} Q_{j_k}| \geq r+s'-2t > r-t$. Therefore, by pigeonhole principle, there exists an index $k \in [r] \setminus \{j_1, \ldots, j_t\}$ such that $|Q_k| \geq 2$. Set $j_{t+1}=k$.
        \paragraph{\boldmath Case 4. $|\cup_{k \in [t]} Q_{j_k}| = 2t$ and $|\cup_{k \in [t]} Q_{j_k} \cap \{1, \ldots, r\}| = t$:} If there exists $k \in [r] \setminus \{j_1, \ldots, j_t\}$, such that $|Q_k \cap [r]| \geq 2$, set $j_{t+1}=k$. Otherwise, for all $k \in [r] \setminus \{j_1, \ldots, j_t\}$, $|Q_k \cap [r]|=1$ since $|\cup_{k \in [t]} Q_{j_k} \cap \{1, \ldots, r\}| = t$ and $|\cup_{k \in [r]} Q_{j_k} \cap \{1, \ldots, r\}| = r$. Set $j_{t+1}=k$ for any $k \in [r] \setminus \{j_1, \ldots, j_t\}$, such that $|Q_k \cap S'| \geq 1$. Since $|\cup_{j \in [r]} Q_j \cap S'| \geq s' > t$, such $k$ exists.
    \end{proof}
Now we prove Claim \ref{claim-2}.
\begin{claim}\label{claim-2}
    $v_i(\{r-s'+1, \ldots, r\} \cup \{2n-\ell - s'+1, \ldots, 2n-\ell\}) \leq s'$.
\end{claim}
\begin{proof}
    Let $Q^1$ be the set of $s'$ most valuable goods in $\cup_{k \in [s']} Q_{j_k}$ and let $Q^2$ be the set of $s'$ least valuable goods in $\cup_{k \in [s']} Q_{j_k}$. Since $|\cup_{k \in [s']} Q_{j_k}| \geq 2s'$, $Q^1 \cap Q^2 = \emptyset$. Also, $|\cup_{k \in [s']} Q_{j_k} \cap \{1, \ldots, r\}| \geq s'$. Thus, $v_i(Q^1) \geq v_i(\{r-s'+1, \ldots, r\})$. Moreover, $v_i(Q^2) \geq v_i(\{2n-\ell - s'+1, \ldots, 2n-\ell\})$. Hence,
    \begin{align*}
        s' &= \sum_{k \in [s']} v_i(P_{j_k}) \notag\\
        &\geq \sum_{k \in [s']} v_i(Q_{j_k}) \notag \\
        &\geq v_i(\{r-s'+1, \ldots, r\} \cup \{2n-\ell - s'+1, \ldots, 2n-\ell\}).
    \end{align*}
\end{proof}
Note that in case $s' = s$, Claim \ref{claim-2} implies Lemma \ref{difficult-bound}. Therefore, from now on, we assume $s' = |S'| < s$.
\begin{claim}\label{claim-22}
    $v_i(\{r-s'+1, \ldots, r\}) + v_i(S') \leq s'$.
\end{claim}
\begin{proof}
    The proof is similar to the proof of Claim \ref{claim-2}. Let $Q^1$ be the set of $s'$ most valuable goods in $\cup_{k \in [s']} Q_{j_k}$ and let $Q^2$ be the set of $s'$ least valuable goods in $\cup_{k \in [s']} Q_{j_k}$. Since $|\cup_{k \in [s']} Q_{j_k}| \geq 2s'$, $Q^1 \cap Q^2 = \emptyset$. Also, $|\cup_{k \in [s']} Q_{j_k} \cap \{1, \ldots, r\}| \geq s'$. Thus, $v_i(Q^1) \geq v_i(\{r-s'+1, \ldots, r\})$. Moreover, $v_i(Q^2) \geq v_i(S')$ since $s' = |S'|$. Hence,
    \begin{align*}
        s' &= \sum_{k \in [s']} v_i(P_{j_k}) \notag\\
        &\geq \sum_{k \in [s']} v_i(Q_{j_k}) \notag \\
        &\geq v_i(\{r-s'+1, \ldots, r\} \cup S')\\
        &= v_i(\{r-s'+1, \ldots, r\}) + v_i(S').
    \end{align*}
\end{proof}
\begin{claim}\label{claim-3}
    $v_i (\{2n-\ell-3s+2s'+1, \ldots, 2n-\ell\}) - v_i(S') \leq s-s'.$
\end{claim}
\begin{proof}
    Note that by definition of $S'$, the $2n-\ell-r-s'=8n/3-2r+s-s'$ goods in $\{r+1, \ldots, 2n-\ell\} \setminus S'$ are in $P_{r+1} \cup \ldots \cup P_{4n/3}$. Now for $j \in [4n/3 - r]$, let $R_j = P_{j+r} \cap \{r+1, \ldots, 2n-\ell\} \setminus S'$. Assume $|R_{j_1}| \geq \ldots \geq |R_{j_{4n/3-r}}|$. We prove
    \begin{align}
        \sum_{k \leq s-s'} |R_{j_k}| \geq 3(s-s'). \label{R-bound}
    \end{align}
    If $|R_{j_{s-s'+1}}| \geq 3$, Inequality \eqref{R-bound} holds. Otherwise, we have
    \begin{align*}
        \frac{8n}{3}-2r+s-s' &= \sum_{k \in [4n/3-r]} |R_{j_k}| \\
        &= \sum_{k \leq s-s'} |R_{j_k}| + \sum_{s-s' < k \leq 4n/3-r} |R_{j_k}| \\
        &\leq \sum_{k \leq s-s'} |R_{j_k}| + 2(\frac{4n}{3}-r-s+s'). \tag{$|R_{j_k}| \leq 2$ for $k > s-s'$}
    \end{align*}
    Thus, $\sum_{k \in [s-s']} |R_{j_k}| \geq 3(s-s')$. We have
    \begin{align*}
        s-s' &= \sum_{k \in [s-s']} v_i(P_{j_k+r}) \notag\\
        &\geq \sum_{k \in [s-s']} v_i(R_{j_k}) \notag\\
        &\geq v_i (\{2n-\ell-3s+2s'+1, \ldots, 2n-\ell\}) - v_i(S'). \tag{$\lvert \cup_{k \in [s-s']} R_{j_k} \rvert \geq 3(s-s')$ and $\lvert S' \rvert =s'$}
    \end{align*}
\end{proof}
Claims \ref{claim-22} and \ref{claim-3} imply Lemma \ref{difficult-bound}.

\paragraph{Recap:} To show that a $1$-out-of-$(4n/3)$ MMS allocation exists, it suffices to prove that we never run out of goods for bag-filling in Algorithm \ref{algo}. Towards contradiction, we assumed that the algorithm stops before agent $i$ receives a bundle. By Observation \ref{less-than-one}, a bag with a value less than $1$ for agent $i$ exists. Let $\ell^*$ be the smallest such that $v_i(B_{\ell^*+1}) < 1$. In Section \ref{negative}, we reached a contradiction assuming $v_i(2n-\ell^*)<1/3$ and proved Theorem \ref{contradict-1}. In Section \ref{positive}, we reached a contradiction assuming $v_i(2n-\ell^*) \geq 1/3$ and proved Theorem \ref{contradict-2}. Therefore, no such agent $i$ exists, and all agents receive a bag by the end of Algorithm \ref{algo}. Theorem \ref{thm:main} follows.

%% file: technical.tex
\subsection{Technical Overview}\label{sec:tec}
For $\alpha$-MMS problem, the algorithms for $\alpha\ge 3/4$~\cite{ghodsi2018fair,garg2021improved,akrami2023simplification} utilize the two-phase approach: \emph{valid reductions} and \emph{bag filling}. In a valid reduction, the instance is reduced by removing an agent $a$ and a subset of goods $S$ such that $v_a(S)\ge \alpha$, and the MMS values of the remaining agents do not decrease.
The valid reduction phase is crucial for the bag filling to work in the analysis of these algorithms. However, it is not clear how to define valid reductions in the case of $1$-out-of-$d$ MMS because $d$ is not the same as the number of agents $n$. Therefore, we only use bag filling in our algorithm, which makes its analysis quite involved and entirely different than from $\alpha$-MMS algorithms.

\begin{algorithm}[htb]
    \caption{$1$-out-of-$4\ceil{n/3}$ MMS}
    \label{algo}
    \textbf{Input:} Ordered $4\ceil{n/3}$-normalized instance $\mathcal{I} = (N, [m], \mathcal{V})$.\\
    \textbf{Output:} Allocation $\hat{B} = (\hat{B}_1, \ldots, \hat{B}_n)$.
    \begin{algorithmic}[1]
    \For{$k \in [n]$}   \Comment{Initialization}
        \State $B_k = \{k, 2n-k+1\}$
    \EndFor
    \State $j \leftarrow 2n+1$
    \For{$k \in [n]$} \Comment{Bag-filling}
        \While{$\nexists i \in N$ s.t. $v_i(B_k) \geq 1$}
            \State $B_k \leftarrow B_k \cup \{j\}$
            \State $j \leftarrow j+1$
            \If{$j>m$}
                \State Terminate
            \EndIf
        \EndWhile
        \State Let $i \in N$ be s.t. $v_i(B_k) \geq 1$
        \State $\hat{B}_i \leftarrow B_k$
        \State $N \leftarrow N \setminus \{i\}$
    \EndFor
    \State $\hat{B}_n \leftarrow \hat{B}_n \cup (M \setminus [j])$
    \State \Return $\hat{B}$
    \end{algorithmic}
\end{algorithm}

The algorithm is described in Algorithm~\ref{algo}. Given an ordered $d$-normalized instance, we initialize $n$ bags (one for each agent) with the first $2n$ (highest valued) goods as follows.
\begin{equation}
    \label{eq:B_i}
    B_k := \{k , 2n-k+1\} \text{ for } k\in [n].
\end{equation}

See Figure \ref{fig:bags} for a better intuition. Then, we do bag-filling. That is, at each round $j$, we keep adding goods in decreasing values to the bag $B_j$ until some agent with no assigned bag values it at least $1$ (recall that $1$-out-of-$d$ MMS value of each agent is 1 in a $d$-normalized instance). Then, we allocate it to an arbitrary such agent.
We note that in contrast to \cite{garg2021improved,akrami2023simplification}, in the bag-filling phase we do not add arbitrary goods to the bags but we add the goods in the decreasing order of their values.

To prove that the output of Algorithm~\ref{algo} is $1$-out-of-$d$ MMS, it is sufficient to prove that we never run out of goods in any round or, equivalently, each agent receives a bag in some round. Towards contradiction, assume that agent $i$ does not receive a bag and the algorithm terminates. It can be easily argued that agent $i$'s value for at least one of the initial bags $\{B_1, \ldots, B_n\}$ must be strictly less than $1$. Let $\ell^*$ be the smallest such that $v_i(B_{\ell^*})<1$. We consider two cases based on the value of $v_i(2n-\ell^*)$. In Section \ref{positive}, we reach a contradiction assuming $v_i(2n-\ell^*) \geq 1/3$ and in Section \ref{negative}, we reach a contradiction assuming $v_i(2n-\ell^*)<1/3$.

Let $\hat{B}_j$ denote the $j$-th bag at the end of the algorithm. The overall idea is to categorize the bags into different groups and prove an upper bound on the value of each bag $(\hat{B}_j)$ for agent $i$ depending on which group it belongs to. Since $v_i(M)=n$ due to the instance being $d$-normalized, we get upper and lower bounds on the size of the groups. For example, if we know that for all bags $\hat{B}_j$ in a certain group $v_i(\hat{B}_j)<1$, we get the trivial upper bound of $n-1$ on the size of this group since $n = v_i(M) = \sum_{j \in [n]}v_i(\hat{B}_j)$.

Unfortunately, upper bounding the value of the bags is not enough to reach a contradiction in all cases. However, for these cases, we have upper and lower bounds on the size of each group, and in general, we show several additional properties to make it work.
For example, we obtain nontrivial upper bounds on the values of certain subsets of goods using the fact that all bundles in a $d$-MMS partition of agent $i$ have value $1$ (see Lemmas \ref{expowerful} and \ref{difficult-bound}).

%% file: figs/B-bags.tikz
\begin{tikzpicture}
[scale=1,
 good/.style={circle, draw=black, thick, minimum size=30pt},
]

\draw[black, very thick] (-0.4-0.25,0.8) rectangle (0.4+0.25,3.4);

\node[good]      at (0,2.75)      {$\scriptstyle{2n}$};
\node[good]      at (0,1.5)      {$\scriptstyle{1}$};

\node at (0, 0.5) {$\scriptstyle{B_1}$};

\filldraw[color=black!60, fill=black!5, thick](1.6, 2) circle (0.02);
\filldraw[color=black!60, fill=black!5, thick](1.7, 2) circle (0.02);
\filldraw[color=black!60, fill=black!5, thick](1.8, 2) circle (0.02);

\draw[black, very thick] (-1.4-0.25 +4.5,0.8) rectangle (-0.6+0.25 +4.5,3.4);

\node[good, scale=0.75]      at (0 +3.5,2.75)      {$\scriptstyle{2n-k+1}$};
\node[good]      at (0 +3.5,1.5)      {$\scriptstyle{k}$};

\node at (0 +3.5, 0.5) {$\scriptstyle{B_k}$};

\filldraw[color=black!60, fill=black!5, thick](0.1+5, 2) circle (0.02);
\filldraw[color=black!60, fill=black!5, thick](0.2+5, 2) circle (0.02);
\filldraw[color=black!60, fill=black!5, thick](0.3+5, 2) circle (0.02);

\draw[black, very thick] (-0.4-5.25 +12,0.8) rectangle (0.4+0.25-5 +12,3.4);

\node[good]      at (0 +7,2.75)      {$\scriptstyle{n+1}$};
\node[good]      at (0 +7,1.5)      {$\scriptstyle{n}$};

\node at (0 +7, 0.5) {$\scriptstyle{B_n}$};

\end{tikzpicture}

%% file: figs/k-picture.tikz
\begin{tikzpicture}
[scale=1,
 good/.style={circle, draw=black, thick, minimum size=30pt},
]

\draw[blue, fill=blue!5, very thick] (-0.4-0.35,3.5) rectangle (0.26+4,0.3);
\node at (1.8, 0) {\textcolor{blue}{$\scriptstyle{A^+}$}};

\draw[black, very thick] (-0.4-0.25,0.8) rectangle (0.4+0.25,3.4);

\node[good]      at (0,2.75)      {$\scriptstyle{2n}$};
\node[good]      at (0,1.5)      {$\scriptstyle{1}$};

\node at (0, 0.5) {$\scriptstyle{B_1}$};

\filldraw[color=black!60, fill=black!5, thick](1.6, 2) circle (0.02);
\filldraw[color=black!60, fill=black!5, thick](1.7, 2) circle (0.02);
\filldraw[color=black!60, fill=black!5, thick](1.8, 2) circle (0.02);

\draw[black, very thick] (-1.4-0.25 +4.5,0.8) rectangle (-0.6+0.25 +4.5,3.4);

\node[good, scale=0.86]      at (0 +3.5,2.75)      {$\scriptstyle{2n-\ell^*+1}$};
\node[good]      at (0 +3.5,1.5)      {$\scriptstyle{\ell^*}$};

\node at (0 +3.5, 0.5) {$\scriptstyle{B_{\ell^*}}$};

\draw[blue, fill=blue!5, very thick] (0.35+4,3.5) rectangle (4.75+5,0.3);
\node at (7, 0) {\textcolor{blue}{$\scriptstyle{\{B_{\ell^*+1}, \ldots, B_s\}\subset A^1 \cup A^2}$}};

\draw[black, very thick] (-2.3-0.25 +7,0.8) rectangle (-1.5+0.25 +7,3.4);

\node[good,scale=0.95]      at (0 +5.1,2.75)      {$\scriptstyle{2n-\ell^*}$};
\node[good]      at (0 +5.1,1.5)      {$\scriptstyle{\ell^*+1}$};

\node at (0 +5.1, 0.5) {$\scriptstyle{B_{\ell^*+1}}$};

\filldraw[color=black!60, fill=black!5, thick](2+4.9, 2) circle (0.02);
\filldraw[color=black!60, fill=black!5, thick](2.1+4.9, 2) circle (0.02);
\filldraw[color=black!60, fill=black!5, thick](2.2+4.9, 2) circle (0.02);

\draw[black, very thick] (-3.4-0.25 +12,0.8) rectangle (-2.6+0.25 +12,3.4);

\node[good, scale=0.9]      at (0 +9,2.75)      {$\scriptstyle{2n-s+1}$};
\node[good]      at (0 +9,1.5)      {$\scriptstyle{s}$};

\node at (0 +9, 0.5) {$\scriptstyle{B_s}$};

\draw[blue, fill=blue!5, very thick] (4.85+5,3.5) rectangle (7.75+5,0.3);
\node at (11.2, 0) {\textcolor{blue}{$\scriptstyle{\{B_{s+1}, \ldots, B_n\}\subset A^2}$}};

\filldraw[color=black!60, fill=black!5, thick](5+5.3, 2) circle (0.02);
\filldraw[color=black!60, fill=black!5, thick](5.1+5.3, 2) circle (0.02);
\filldraw[color=black!60, fill=black!5, thick](5.2+5.3, 2) circle (0.02);

\draw[black, very thick] (-0.4-0.25 +12,0.8) rectangle (0.4+0.25 +12,3.4);

\node[good]      at (0 +12,2.75)      {$\scriptstyle{n+1}$};
\node[good]      at (0 +12,1.5)      {$\scriptstyle{n}$};

\node at (0 +12, 0.5) {$\scriptstyle{B_n}$};

\end{tikzpicture}

%% file: figs/bag-filling.tikz
\begin{tikzpicture}
[scale=1,
 good/.style={circle, draw=black, thick, minimum size=30pt},
]

\draw[blue, fill=blue!5, very thick] (-0.4-0.35-0.4,3.5) rectangle (0.4+2.6-0.4,2.1);

\draw[red, fill=red!5, very thick] (-0.4-0.35-0.4,2.07) rectangle (0.4+2.6-0.4,0.75);

\draw[black, very thick] (-0.4-0.25-0.4,0.8) rectangle (0.4+0.25-0.4,3.4);

\node[good]      at (0-0.4,2.75)      {$\scriptstyle{2n}$};
\node[good]      at (0-0.4,1.5)      {$\scriptstyle{1}$};

\node at (0-0.4, 0.5) {$\scriptstyle{B_1}$};

\filldraw[color=black!60, fill=black!5, thick](1-0.4, 2) circle (0.02);
\filldraw[color=black!60, fill=black!5, thick](1.1-0.4, 2) circle (0.02);
\filldraw[color=black!60, fill=black!5, thick](1.2-0.4, 2) circle (0.02);

\draw[black, very thick] (-0.4-0.25 +2.25-0.4,0.8) rectangle (0.4+0.25 +2.25-0.4,3.4);

\node[good]      at (0 +2.25-0.4,2.75)      {.};
\node[good]      at (0 +2.25-0.4,1.5)      {.};

\node at (0 +2.25-0.4, 0.5) {$\scriptstyle{B_{2\ell+t+2\ell^*-2n/3}}$}; 

\filldraw[color=black!60, fill=black!5, thick](1+1.8, 2) circle (0.02);
\filldraw[color=black!60, fill=black!5, thick](1.1+1.8, 2) circle (0.02);
\filldraw[color=black!60, fill=black!5, thick](1.2+1.8, 2) circle (0.02);

\draw[blue, fill=blue!5, very thick] (-0.4-0.25 +3.9,0.2) rectangle (0.4+0.25 +6.1,3.5);

\draw[black, very thick] (-0.4-0.25 +4,0.8) rectangle (0.4+0.25 +4,3.4);

\node[good]      at (0 +4,2.75)      {$\scriptstyle{2n-t}$};
\node[good]      at (0 +4,1.5)      {$\scriptstyle{t+1}$};

\node at (0 +4, 0.5) {$\scriptstyle{B_{t+1}}$};

\filldraw[color=black!60, fill=black!5, thick](1+3.9, 2) circle (0.02);
\filldraw[color=black!60, fill=black!5, thick](1.1+3.9, 2) circle (0.02);
\filldraw[color=black!60, fill=black!5, thick](1.2+3.9, 2) circle (0.02);

\draw[black, very thick] (-0.4-0.25 +6,0.8) rectangle (0.4+0.25 +6,3.4);

\node[good, scale=0.8]      at (0 +6,2.75)      {$\scriptstyle{2n-\ell^*+1}$};
\node[good]      at (0 +6,1.5)      {$\scriptstyle{\ell^*}$};

\node at (0 +6, 0.5) {$\scriptstyle{B_{\ell^*}}$};

\draw[blue, fill=blue!5, very thick] (-0.4-0.25 +7.5,3.55) rectangle (0.4+0.25 +9.7,4.8);

\draw[black, very thick] (-0.4-0.25 +7.5,0.8) rectangle (0.4+0.25 +7.5,3.4);

\node[good,scale=0.9]      at (0 +7.5,2.75)      {$\scriptstyle{2n-\ell^*}$};
\node[good]      at (0 +7.5,1.5)      {$\scriptstyle{\ell^*+1}$};

\node[good]      at (0 +7.5,4.18)      {$\scriptstyle{2n+1}$};

\node at (0 +7.5, 0.5) {$\scriptstyle{B_{\ell^*+1}}$};

\filldraw[color=black!60, fill=black!5, thick](1+7.5, 2) circle (0.02);
\filldraw[color=black!60, fill=black!5, thick](1.1+7.5, 2) circle (0.02);
\filldraw[color=black!60, fill=black!5, thick](1.2+7.5, 2) circle (0.02);

\draw[black, very thick] (-0.4-0.25 +9.7,0.8) rectangle (0.4+0.25 +9.7,3.4);

\node[good, scale=0.6]      at (0 +9.7,2.75) {$2n-s+1$};
\node[good]      at (0 +9.7,1.5)      {$\scriptstyle{s}$};

\node[good, scale=0.95]      at (0 +9.7,4.18)      {$\scriptstyle{\frac{8n}{3} + \ell}$};

\node at (0 +9.7, 0.5) {$\scriptstyle{B_s}$};

\filldraw[color=black!60, fill=black!5, thick](1+9.6, 2) circle (0.02);
\filldraw[color=black!60, fill=black!5, thick](1.1+9.6, 2) circle (0.02);
\filldraw[color=black!60, fill=black!5, thick](1.2+9.6, 2) circle (0.02);

\draw[black, very thick] (-0.4-0.25 +12,0.8) rectangle (0.4+0.25 +12,3.4);

\node[good]      at (0 +12,2.75)      {$\scriptstyle{n+1}$};
\node[good]      at (0 +12,1.5)      {$\scriptstyle{n}$};

\node at (0 +12, 0.5) {$\scriptstyle{B_n}$};

\end{tikzpicture}

%% file: figs/Lemma-Fig.tikz
\begin{tikzpicture}
[scale=1,
 good/.style={circle, draw=black, thick, minimum size=30pt},
]

\draw[blue, fill=blue!5, very thick] (-0.4-0.35,3.5) rectangle (0.4+1.1,2.1);

\draw[black, very thick] (-0.4-0.25,0.8) rectangle (0.4+0.25,3.4);

\node[good]      at (0,2.75)      {$\scriptstyle{2n}$};
\node[good]      at (0,1.5)      {$\scriptstyle{1}$};

\node at (0, 0.5) {$\scriptstyle{B_1}$};

\filldraw[color=black!60, fill=black!5, thick](1, 2) circle (0.02);
\filldraw[color=black!60, fill=black!5, thick](1.1, 2) circle (0.02);
\filldraw[color=black!60, fill=black!5, thick](1.2, 2) circle (0.02);

\draw[red, fill=red!5, very thick] (0.4+1.15,3.5) rectangle (5,2.1);

\draw[black, very thick] (-0.4-0.25 +2.25,0.8) rectangle (0.4+0.25 +2.25,3.4);

\node[good]      at (0 +2.25,2.75)      {};
\node[scale=0.8]      at (0 +2.25,2.85)      {$\scriptstyle{\frac{8n}{3} - 2\ell}$};
\node[scale=0.8]      at (0 +2.25,2.55)      {$\scriptstyle{- 2t - \ell^*}$};

\node[good]      at (0 +2.25,1.5)     {};

\node[scale=0.8] at (0 +2.25,1.6) {$\scriptstyle{2\ell + 2t + \ell^*}$};
\node[scale=0.8] at (0 +2.25,1.3) {$\scriptstyle{-\frac{2n}{3} + 1}$};

\node[scale=0.9] at (0 +2.25, 0.4) {$\scriptstyle{B_{2\ell + 2t + \ell^* -\frac{2n}{3} + 1}}$};

\filldraw[color=black!60, fill=black!5, thick](1+2.2, 2) circle (0.02);
\filldraw[color=black!60, fill=black!5, thick](1.1+2.2, 2) circle (0.02);
\filldraw[color=black!60, fill=black!5, thick](1.2+2.2, 2) circle (0.02);

\draw[black, very thick] (-0.4-0.25 +4.25,0.8) rectangle (0.4+0.25 +4.25,3.4);

\node[good]      at (0 +4.25,2.75)      {};
\node[scale=0.8]      at (0 +4.25,2.83)      {$\scriptstyle{\frac{8n}{3} - 2\ell -t}$};
\node[scale=0.8]      at (0 +4.25,2.53)      {$\scriptstyle{-2\ell^*+1}$};

\node[good]      at (0 +4.25,1.5)      {};
\node[scale=0.8]      at (0 +4.25,1.6)      {$\scriptstyle{2\ell+t+2\ell^*}$};
\node[scale=0.8]      at (0 +4.25,1.3)      {$\scriptstyle{-\frac{2n}{3}}$};

\node[scale=0.9] at (0 +4.25, 0.4) {$\scriptstyle{B_{2\ell+t+2\ell^*-\frac{2n}{3}}}$};

\filldraw[color=black!60, fill=black!5, thick](1+4.2, 2) circle (0.02);
\filldraw[color=black!60, fill=black!5, thick](1.1+4.2, 2) circle (0.02);
\filldraw[color=black!60, fill=black!5, thick](1.2+4.2, 2) circle (0.02);

\draw[red, fill=red!5, very thick] (-0.4-0.25 +6.4,2.13) rectangle (0.4+0.25 +7.6,3.5);

\draw[blue, fill=blue!5, very thick] (-0.4-0.25 +6.4,0.7) rectangle (0.4+0.25 +7.6,2.07);

\draw[black, very thick] (-0.4-0.25 +6.5,0.8) rectangle (0.4+0.25 +6.5,3.4);

\node[good]      at (0 +6.5,2.75)      {$\scriptstyle{2n-t}$};
\node[good]      at (0 +6.5,1.5)      {$\scriptstyle{t+1}$};

\node[scale=0.9] at (0 +6.5, 0.4) {$\scriptstyle{B_{t+1}}$};

\filldraw[color=black!60, fill=black!5, thick](1+6.6, 2) circle (0.02);
\filldraw[color=black!60, fill=black!5, thick](1.1+6.6, 2) circle (0.02);
\filldraw[color=black!60, fill=black!5, thick](1.2+6.6, 2) circle (0.02);

\draw[blue, fill=blue!5, very thick] (-0.4-0.25 +9,3.55) rectangle (0.4+0.25 +11.5,4.8);

\draw[black, very thick] (-0.4-0.25 +9,0.8) rectangle (0.4+0.25 +9,3.4);

\node[good]      at (0 +9,2.75)      {$\scriptstyle{2n-\ell^*}$};
\node[good]      at (0 +9,1.5)      {$\scriptstyle{\ell^*+1}$};

\node[good]      at (0 +9,4.18)      {$\scriptstyle{2n+1}$};

\node at (0 +9, 0.5) {$\scriptstyle{B_{\ell^*+1}}$};

\filldraw[color=black!60, fill=black!5, thick](1+9.1, 2) circle (0.02);
\filldraw[color=black!60, fill=black!5, thick](1.1+9.1, 2) circle (0.02);
\filldraw[color=black!60, fill=black!5, thick](1.2+9.1, 2) circle (0.02);

\draw[black, very thick] (-0.4-0.25 +11.5,0.8) rectangle (0.4+0.25 +11.5,3.4);

\node[good, scale=0.53]      at (0 +11.5,2.75) {$2n-s+1$};
\node[good]      at (0 +11.5,1.5)      {$\scriptstyle{s}$};

\node[good]      at (0 +11.5,4.18)      {$\scriptstyle{\frac{8n}{3} + \ell}$};

\node at (0 +11.5, 0.5) {$\scriptstyle{B_s}$};

\filldraw[color=black!60, fill=black!5, thick](1+11.5, 2) circle (0.02);
\filldraw[color=black!60, fill=black!5, thick](1.1+11.5, 2) circle (0.02);
\filldraw[color=black!60, fill=black!5, thick](1.2+11.5, 2) circle (0.02);

\end{tikzpicture}

%% file: figs/r-l-fig.tikz
\begin{tikzpicture}
[scale=1,
 good/.style={circle, draw=black, thick, minimum size=30pt},
]

\draw[blue, fill=blue!5, very thick] (-0.4-0.35,3.5) rectangle (0.26+4,0.3);
\node[scale=0.7] at (1.8, 0) {\textcolor{blue}{$v_i(\hat{B}_j) \leq 4/3-x$}};
\node[scale=0.7] at (1.8, -0.5) {\textcolor{blue}{by Corollary \ref{cor-upper-left}}};

\draw[blue, fill=blue!5, very thick] (-0.4-0.35+8.5,3.5) rectangle (0.26+4+8.5,0.3);
\node[scale=0.7] at (1.8+8.5, 0) {\textcolor{blue}{$v_i(\hat{B}_j) \leq \max(4/3-x, 4/3-2y)$}};
\node[scale=0.7] at (1.8+8.5, -0.5) {\textcolor{blue}{by Corollary \ref{cor-upper-right}}};

\draw[red, fill=red!5, very thick] (-0.4-0.35+5.1,3.5) rectangle (0.26+4+3.4,0.3);
\node[scale=0.7] at (1.8+4.3, 0) {\textcolor{red}{$1 \leq v_i(\hat{B}_j) < 1+x+y$}};
\node[scale=0.7] at (1.8+4.3, -0.5) {\textcolor{red}{by Observation \ref{middle-block}}};

\draw[black, very thick] (-0.4-0.25,0.8) rectangle (0.4+0.25,3.4);

\node[good]      at (0,2.75)      {$\scriptstyle{2n}$};
\node[good]      at (0,1.5)      {$\scriptstyle{1}$};

\node at (0, 0.5) {$\scriptstyle{B_1}$};

\filldraw[color=black!60, fill=black!5, thick](1.6, 2) circle (0.02);
\filldraw[color=black!60, fill=black!5, thick](1.7, 2) circle (0.02);
\filldraw[color=black!60, fill=black!5, thick](1.8, 2) circle (0.02);

\draw[black, very thick] (-1.4-0.25 +4.5,0.8) rectangle (-0.6+0.25 +4.5,3.4);

\node[good, scale=0.8]      at (0 +3.5,2.75)      {$\scriptstyle{2n-\ell+1}$};
\node[good]      at (0 +3.5,1.5)      {$\scriptstyle{\ell}$};

\node at (0 +3.5, 0.5) {$\scriptstyle{B_\ell}$};

\filldraw[color=black!60, fill=black!5, thick](1+4.9, 2) circle (0.02);
\filldraw[color=black!60, fill=black!5, thick](1.1+4.9, 2) circle (0.02);
\filldraw[color=black!60, fill=black!5, thick](1.2+4.9, 2) circle (0.02);

\draw[black, very thick] (-3.9-0.25 +12,0.8) rectangle (-3.1+0.25 +12,3.4);

\node[good]      at (-0.5 +9,2.75)      {$\scriptstyle{2n-r}$};
\node[good]      at (-0.5 +9,1.5)      {$\scriptstyle{r+1}$};

\node at (-0.5 +9, 0.5) {$\scriptstyle{B_{r+1}}$};

\filldraw[color=black!60, fill=black!5, thick](5+5.1, 2) circle (0.02);
\filldraw[color=black!60, fill=black!5, thick](5.1+5.1, 2) circle (0.02);
\filldraw[color=black!60, fill=black!5, thick](5.2+5.1, 2) circle (0.02);

\draw[black, very thick] (-0.4-0.25 +12,0.8) rectangle (0.4+0.25 +12,3.4);

\node[good]      at (0 +12,2.75)      {$\scriptstyle{n+1}$};
\node[good]      at (0 +12,1.5)      {$\scriptstyle{n}$};

\node at (0 +12, 0.5) {$\scriptstyle{B_n}$};

\end{tikzpicture}

%% file: figs/Lemma-Fig-1.tikz
\begin{tikzpicture}
[scale=1,
 good/.style={circle, draw=black, thick, minimum size=30pt},
]

\draw[black, very thick] (-0.4-0.25,0.8) rectangle (0.4+0.25,3.4);

\node[good]      at (0,2.75)      {$\scriptstyle{2n}$};
\node[good]      at (0,1.5)      {$\scriptstyle{1}$};

\node at (0, 0.5) {$\scriptstyle{B_1}$};

\filldraw[color=black!60, fill=black!5, thick](1, 2) circle (0.02);
\filldraw[color=black!60, fill=black!5, thick](1.1, 2) circle (0.02);
\filldraw[color=black!60, fill=black!5, thick](1.2, 2) circle (0.02);

\draw[blue, fill=blue!5, very thick] (0.4+1.12,2.1) rectangle (4.98,3.33);

\draw[black, very thick] (-0.4-0.25 +2.25,0.8) rectangle (0.4+0.25 +2.25,3.4);

\node[good]      at (0 +2.25,2.75)      {$\scriptstyle{2n-\ell}$};

\node[good]      at (0 +2.25,1.5)     {$\scriptstyle{\ell+1}$};

\node at (0 +2.25, 0.5) {$\scriptstyle{B_{\ell+1}}$};

\filldraw[color=black!60, fill=black!5, thick](1+2.2, 2) circle (0.02);
\filldraw[color=black!60, fill=black!5, thick](1.1+2.2, 2) circle (0.02);
\filldraw[color=black!60, fill=black!5, thick](1.2+2.2, 2) circle (0.02);

\draw[black, very thick] (-0.4-0.25 +4.25,0.8) rectangle (0.4+0.25 +4.25,3.4);

\node[scale=0.75]  at (0 +4.25,2.85)      {$\scriptstyle{2n-\ell-3s}$};
\node[scale=0.75]  at (0 +4.25,2.63)      {$\scriptstyle{+2s'+1}$};
\node[good]      at (0 +4.25,2.75)      {};
\node[good]      at (0 +4.25,1.5)      {};
\node[scale=0.75]  at (0 +4.25,1.5)      {$\scriptstyle{\ell+3s-2s'}$};

\node at (0 +4.25, 0.5) {$\scriptstyle{B_{\ell+3s-2s'}}$};

\filldraw[color=black!60, fill=black!5, thick](1+4.2, 2) circle (0.02);
\filldraw[color=black!60, fill=black!5, thick](1.1+4.2, 2) circle (0.02);
\filldraw[color=black!60, fill=black!5, thick](1.2+4.2, 2) circle (0.02);

\draw[blue, fill=blue!5, very thick] (-0.4-0.25 +6.42,0.9) rectangle (0.4+0.25 +9.08,2.13);

\draw[black, very thick] (-0.4-0.25 +6.5,0.8) rectangle (0.4+0.25 +6.5,3.4);

\node[good]      at (0 +6.5,2.75)      {$\scriptstyle{}$};
\node[scale=0.8] at (0 +6.5,2.75)      {$\scriptstyle{2n-r+s'}$};

\node[good]      at (0 +6.5,1.5)      {};
\node[scale=0.9] at (0 +6.5,1.5)      {$\scriptstyle{r-s'+1}$};

\node at (0 +6.5, 0.5) {$\scriptstyle{B_{r-s'+1}}$};

\filldraw[color=black!60, fill=black!5, thick](1+6.6, 2) circle (0.02);
\filldraw[color=black!60, fill=black!5, thick](1.1+6.6, 2) circle (0.02);
\filldraw[color=black!60, fill=black!5, thick](1.2+6.6, 2) circle (0.02);

\draw[black, very thick] (-0.4-0.25 +9,0.8) rectangle (0.4+0.25 +9,3.4);

\node[good, scale=0.75]      at (0 +9,2.75)      {$\scriptstyle{2n-r+1}$};
\node[good]      at (0 +9,1.5)      {$\scriptstyle{r}$};

\node at (0 +9, 0.5) {$\scriptstyle{B_{r}}$};

\filldraw[color=black!60, fill=black!5, thick](1+9.1, 2) circle (0.02);
\filldraw[color=black!60, fill=black!5, thick](1.1+9.1, 2) circle (0.02);
\filldraw[color=black!60, fill=black!5, thick](1.2+9.1, 2) circle (0.02);

\draw[black, very thick] (-0.4-0.25 +11.5,0.8) rectangle (0.4+0.25 +11.5,3.4);

\node[good]      at (0 +11.5,2.75) {$\scriptstyle{n+1}$};
\node[good]      at (0 +11.5,1.5)      {$\scriptstyle{n}$};

\node at (0 +11.5, 0.5) {$\scriptstyle{B_n}$};

\end{tikzpicture}

%% file: hard3.tex
\subsection{Tight Example}\label{sec:tighteg}

We now show that \cref{algo}'s guarantees are almost tight, i.e.,
it cannot give us better than 1-out-of-$\floor{(4n+1)/3}$ MMS.
Note that $\ceil{4n/3} = \floor{(4n+2)/3}$ for all $n \in \mathbb{N}$.

Specifically, we show a fair division instance where \cref{algo}'s output is not
1-out-of-$\floor{(4n-2)/3}$ MMS, even when agents have identical valuations.
This instance is similar to the tight example in \cite{akrami2023simplification},
where they show that a natural class of algorithms (containing \cref{algo})
cannot guarantee a better multiplicative approximation than $(\frac{3}{4} + \frac{3}{8n-4})$-MMS.

\begin{example}
\label{ex:1ood-hard}
Consider a fair division instance with $n$ agents, where $n \ge 2$.
Let $d \defeq \floor{(4n-2)/3} = n + \floor{(n-2)/3}$.
There are $m \defeq 2n+1+3(d-n) \le 3n-1$ goods.
All agents have the same valuation function $u$, where
\[ u(j) \defeq \begin{cases}
\displaystyle \frac{2}{3} - \frac{\ceil{j/2}}{3n} & \textrm{ if } 1 \le j \le 2n
\\ 1/3 & \textrm{ if } 2n < j \le m
\end{cases}. \]
\end{example}

\Cref{fig:1ood-hard} shows that when we run \cref{algo} on \cref{ex:1ood-hard} with $n=5$,
it does not output a 1-out-of-$d$ MMS allocation.
We formally prove this for all $n$:
in \cref{thm:1ood-hard:ord-norm}, we show that \cref{ex:1ood-hard} is ordered and normalized,
and in \cref{thm:1ood-hard:output}, we show that \cref{algo}'s output on \cref{ex:1ood-hard}
gives someone a bundle of value less than 1.
This proves that \cref{algo} does not output a 1-out-of-$d$ MMS allocation on \cref{ex:1ood-hard}.

\begin{figure}[t]
\centering
\begin{subfigure}{0.45\textwidth}
\centering
\input{figs/1ood-algo.tikz}
\caption{\Cref{algo}'s output.}
\label{fig:1ood-algo}
\end{subfigure}
\begin{subfigure}{0.45\textwidth}
\centering
\input{figs/1ood-mms.tikz}
\caption{1-out-of-$d$ MMS partition.}
\label{fig:1ood-mms}
\end{subfigure}
\captionAndDescr{For \cref{ex:1ood-hard} with $n=5$ (and $d = \floor{(4n-2)/3} = 6$),
\cref{algo} outputs an allocation containing a bundle of value $14/15$, but the 1-out-of-$d$ MMS is 1.
(Each cell represents a good, and each column represents a bundle.)}\vspace{-1ex}
\label{fig:1ood-hard}
\end{figure}

\begin{lemma}\label{thm:1ood-hard:ord-norm}
\Cref{ex:1ood-hard} is ordered and $d$-normalized.    
\end{lemma}
\begin{proof}
$u$ is ordered since $u(2n) = u(2n+1) = 1/3$.

Define $M \defeq (M_1, \ldots, M_d)$ as
\[ M_i \defeq \begin{cases}
\{i, 2n-1-i\} & \textrm{ if } 1 \le i \le n-1
\\ \{i+n-1, 2d+n-i, 2d-n+i+1\} & \textrm{ if } n \le i \le d
\end{cases}. \]
To show that $u$ is normalized, we show that $u(M_i) = 1$ for all $i \in [d]$.
For $i \in [n-1]$, we have
\begin{align*}
u(M_i) &= u(i) + u(2n-1-i) = \frac{4}{3} - \frac{\ceil{i/2} + \ceil{(2n-1-i)/2}}{3n}
\\ &= 1 - \frac{\ceil{i/2} + \ceil{-(i+1)/2}}{3n} = 1.
\end{align*}
For $i \in [d] \setminus [n-1]$ and $g \in M_i$, we have $g \ge 2n-1$.
Since $u(j) = 1/3$ for all $j \ge 2n-1$, we get that $u(M_i) = 1$ for $i \in [d] \setminus [n-1]$.
\end{proof}

\begin{lemma}\label{thm:1ood-hard:output}
Let $A \defeq (A_1, \ldots, A_n)$ be \cref{algo}'s output on \cref{ex:1ood-hard}.
Then $u(A_n) = 1 - 1/(3n)$.    
\end{lemma}
\begin{proof}
For $i \in [n]$, let $B_i \defeq \{i, 2n+1-i\}$ be the initial bag. Then
\begin{align*}
u(B_i) &= u(i) + u(2n+1-i) = \frac{4}{3} - \frac{\ceil{i/2} + \ceil{(2n+1-i)/2}}{3n}
\\ &= \frac{4}{3} - \frac{(n+1) + \ceil{i/2} + \ceil{-(i+1)/2}}{3n}
= 1 - \frac{1}{3n}.
\end{align*}
Since $m-2n \le n-1$ goods are left after bag initialization, and each good has value $1/3$,
the first $m-2n$ bags get 1 good each during bag filling,
and the remaining bags do not get any goods.
\end{proof}

%% file: figs/1ood-algo.tikz
\setlength{\cellW}{2.5em}
\setlength{\cellH}{2em}
\begin{tikzpicture}
\foreach \x/\y/\v/\redfrac/\bluefrac in {
    0/0/9/4/0,
    1/0/9/4/0,
    2/0/8/3/1,
    3/0/8/3/1,
    4/0/7/2/2,
    4/1/7/2/2,
    3/1/6/1/3,
    2/1/6/1/3,
    1/1/5/0/4,
    0/1/5/0/4,
    0/2/5/0/4,
    1/2/5/0/4,
    2/2/5/0/4,
    3/2/5/0/4
    } {
\colorlet{mycolor}{rgb:red,\redfrac;blue,\bluefrac}
\draw[fill={mycolor!25!bgColor}] (\x\cellW, \y\cellH) rectangle node[pos=0.5] {\large $\sfrac{\v}{15}$} +(1\cellW, 1\cellH);
}
\foreach \x/\t in {0.5/1,1.5/2,2.5/3,3.5/4,4.5/5} {
    \node[below] at (\x\cellW, 0) {\footnotesize $\t$};
}
\draw[semithick] (0, 0) -- +(5\cellW, 0);
\end{tikzpicture}

%% file: figs/1ood-mms.tikz
\setlength{\cellW}{2.5em}
\setlength{\cellH}{2em}
\begin{tikzpicture}
\foreach \x/\y/\v/\redfrac/\bluefrac in {
    0/0/9/4/0,
    1/0/9/4/0,
    2/0/8/3/1,
    3/0/8/3/1,
    3/1/7/2/2,
    2/1/7/2/2,
    1/1/6/1/3,
    0/1/6/1/3,
    4/0/5/0/4,
    5/0/5/0/4,
    5/1/5/0/4,
    4/1/5/0/4,
    4/2/5/0/4,
    5/2/5/0/4
    } {
\colorlet{mycolor}{rgb:red,\redfrac;blue,\bluefrac}
\draw[fill={mycolor!25!bgColor}] (\x\cellW, \y\cellH) rectangle node[pos=0.5] {\large $\sfrac{\v}{15}$} +(1\cellW, 1\cellH);
}
\foreach \x/\t in {0.5/1,1.5/2,2.5/3,3.5/4,4.5/5,5.5/6} {
    \node[below] at (\x\cellW, 0) {\footnotesize $\t$};
}
\draw[semithick] (0, 0) -- +(6\cellW, 0);
\end{tikzpicture}

%% file: algo1.tex
\section{MMS with Agent Priority Ranking}
\label{sec:algo1}

We give an algorithm for approximate MMS allocations in the priority ranking setting.
We assume \wLoG{} that the input is $n$-normalized and ordered (see \cref{defn:normalized,defn:ordered}).

Our algorithm first performs a few \emph{reduction} operations (phase 1) and then does bag filling (phase 2).
This approach is very similar to that of \cite{garg2021improved}.
We call our algorithm $\RBF$ (abbreviates `Reductions and Bag Filling')
and formally describe it in \cref{algo:rbf}.
We assume that the input to $\RBF$ is a pair $(\Ical, T)$, where
$\Ical \defeq ([n], [m], \Vcal)$ is an $n$-normalized fair division instance
such that $v_i(1) \ge v_i(2) \ge \ldots \ge v_i(m)$ for all $i \in [n]$,
and $T \defeq [\tau_1, \ldots, \tau_n]$ is a sequence of thresholds
such that $1 \ge \tau_1 \ge \ldots \ge \tau_n > 0$.
We say that an agent $i$ \emph{likes} a set $S$ of goods if $v_i(S) \ge \tau_i$.
We want to show that for a reasonable choice of $T$,
$\RBF$ gives each agent a bundle that they like.

For a set $S \subseteq \mathbb{Z}_{\ge 1}$ and an index $j \in [|S|]$,
let $\ordSt(S, j)$ be the $j\Th$ smallest number in $S$
(called the $j\Th$ order statistic of $S$).
For a finite set $J \subset \mathbb{Z}_{\ge 1}$, let
$\ordSt(S, J) \defeq \{\ordSt(S, j): j \in J \textrm{ and } j \le |S|\}$.

The first phase of the algorithm proceeds in multiple rounds.
In each round, we pick a set $S$ of goods and give it to an agent $i$.
The goods $S$ and agent $i$ are then removed from consideration in subsequent rounds.
This operation is called a \emph{reduction}.
Specifically, let $N$ and $M$ be the set of remaining agents and goods,
respectively, at the beginning of a round.
Let $S_1 \defeq \ordSt(M, 1)$, $S_2 \defeq \ordSt(M, \{|N|, |N|+1\})$,
$S_3 \defeq \ordSt(M, \{2|N|-1, 2|N|, 2|N|+1\})$, and $S_4 \defeq \ordSt(M, \{1, 2|N|+1\})$.
Then we find the smallest $k$ such that some agent likes $S_k$,
and the smallest $i$ such that $i$ likes $S_k$. Then we give $S_k$ to agent $i$.
If no such $k$ exists, then phase 1 ends.

\begin{algorithm}[htb]
\caption{$\RBF(\Ical, T)$:
\\ \textbf{Input:} Ordered and $n$-normalized instance $\Ical = ([n], [m], \Vcal)$
    and agent thresholds $T \defeq [\tau_1, \ldots, \tau_n]$.
\\ \textbf{Output:} (Partial) allocation $A = (A_1, \ldots, A_n)$.
}
\label{algo:rbf}
\begin{algorithmic}[1]
\State $N = [n]$
\State $M = [m]$
\While{$|N| > 0$ and $|M| > 0$}
    \State Let $S_1 \defeq \ordSt(M, 1)$.
    \State Let $S_2 \defeq \ordSt(M, \{|N|, |N|+1\})$.
    \State Let $S_3 \defeq \ordSt(M, \{2|N|-1, 2|N|, 2|N|+1\})$.
    \State Let $S_4 \defeq \ordSt(M, \{1, 2|N|+1\})$.
    \State Find the (lexicographically) smallest pair $(k, i) \in [4] \times N$ such that
        $v_i(S_k) \ge \tau_i$. Let $k$ and $i$ be $\Null$ if no such pair exists.
        \Comment{this event is called a type-$k$ reduction.}
    \If{$i$ is not $\Null$}
        \State Give $S_k$ to agent $i$.
            Set $N = N \setminus \{i\}$ and $M = M \setminus S_k$.
    \Else
        \State break
    \EndIf
\EndWhile
\LineComment{$|M| \ge 2|N|$ now (see \cref{thm:m-ge-2n}).}
\State $\bagFillHyp((N, M, \Vcal), T)$
\end{algorithmic}
\end{algorithm}

Then in phase 2, we perform bag filling (\cref{algo:bagFill}) on the remaining instance.
We create $|N|$ bags, where in the $i\Th$ bag, we add the $i\Th$ and $(2n-i+1)\Th$
most valuable goods (so the first $2|N|$ goods are in bags).
Then we repeatedly do the following till all agents receive a bag:
Find the smallest $i$ such that agent $i$ likes some bag, and give that bag to $i$.
If no such $i$ exists (i.e., no agent likes any bag),
add the most valuable remaining good to an arbitrary bag.

\begin{algorithm}[htb]
\caption{$\bagFill(\Ical, T)$
\\ \textbf{Input:} Ordered instance $\Ical = ([n], [m], \Vcal)$ with $m \ge 2n$,
    and agent thresholds $T \defeq [\tau_1, \ldots, \tau_n]$.
\\ \textbf{Output:} (Partial) allocation $A = (A_1, \ldots, A_n)$.
}
\label{algo:bagFill}
\begin{algorithmic}[1]
\For{$i \in [n]$}
    \State $B_i = \{i, 2n+1-i\}$.
    \State $A_i = \emptyset$.
\EndFor
\State $U_G = [m] \setminus [2n]$  \Comment{unassigned goods}
\State $U_A = [n]$  \Comment{unsatisfied agents}
\State $U_B = [n]$  \Comment{unassigned bags}
\While{$U_A \neq \emptyset$}
    \Comment{loop invariant: $|U_A| = |U_B|$}
    \State Let $i$ be the smallest in $U_A$ such that for some $k \in U_B$,
        we have $v_i(B_k) \ge \tau_i$. Let $i = \Null$ otherwise.
    \If{$i$ is not $\Null$}
        \LineComment{assign the $k\Th$ bag to agent $i$:}
        \State $A_i = B_k$
        \State $U_A = U_A \setminus \{i\}$
        \State $U_B = U_B \setminus \{k\}$
    \ElsIf{$U_G \neq \emptyset$}
        \State $g$ = most valuable good in $U_G$
        \State $k$ = arbitrary bag in $U_B$
        \LineComment{assign $g$ to the $k\Th$ bag:}
        \State $B_k = B_k \cup \{g\}$.
        \State $U_G = U_G \setminus \{g\}$
    \Else
        \State \label{alg-line:bagFill:error}\textbf{error}: we ran out of goods.
        \State Allocate the bags in $U_B$ to agents in $U_A$ arbitrarily.
        \State \Return $(A_1, \ldots, A_n)$.
        \LineComment{agents $U_A$ are not satisfied with their bags.}
    \EndIf
\EndWhile
\State \Return $(A_1, \ldots, A_n)$
\end{algorithmic}
\end{algorithm}

Our algorithm ($\RBF$) is the same as that of \cite{garg2021improved}, with a few minor differences:
\begin{enumerate}
\item We assume the input to be $n$-normalized, whereas they do not.
\item In both phases 1 and 2, we assume that when multiple agents like a bundle,
    we give it to the smallest-numbered agent, whereas they break ties arbitrarily.
\item They use the same thresholds for all agents, i.e., $\tau_1 = \ldots = \tau_n$.
\end{enumerate}

The idea behind reductions is to shrink the set of agents to consider:
we want to give each agent a bundle that they like, and since each agent selected for a reduction
obtains a bundle that she likes, we only need to worry about the remaining agents.
Reductions also shrink the set of goods, though, and it is therefore necessary to argue
that the set of remaining goods is enough to satisfy the remaining agents,
i.e., we are not giving away too many goods during reductions.

\cite{garg2021improved} showed this using the idea of \emph{valid reductions}:
a reduction is called \emph{valid} if it preserves the MMS values of the remaining agents.
Formally, they showed that if $\tau_j \le 3/4$ for all $j$,
and if a reduction involves giving $S_k$ to agent $i$, then for all $j \in N \setminus \{i\}$,
we have $\MMS_j^{|N|-1}(M \setminus S_k) \ge \MMS_j^{|N|}(M)$.
Hence, if the initial MMS value of each agent is 1
(which happens if we start with an $n$-normalized instance),
then the MMS value of each remaining agent after phase 1 is at least 1.
Since \cite{garg2021improved} aim to find $3/4$-MMS allocations,
they assume \wLoG{} that the input instance is irreducible, i.e., no reductions are possible.
They show that irreducible instances have nice properties, which they exploit to show that
$\bagFill$ outputs a $3/4$-MMS allocation when its input is irreducible.

Unfortunately, reductions in $\RBF$ are not valid, i.e.,
they can decrease the MMS value of the remaining agents.
Specifically, when $\tau_j > 3/4$ for an agent $j$,
then it can happen that $\MMS^{n-1}_j(M \setminus S_4) < \MMS^n_j(M)$.
Moreover, even for fair division instances where reductions are valid
(e.g., no type-4 reductions take place),
it is unclear how to exploit the validity of reductions in our analysis.
This is because reductions can change the priority rank of agents, e.g.,
the $(n/2)\Th$ most important agent in the original instance may end up becoming
the least important in the reduced instance if $n/2$ reductions take place.
Hence, we need to use other techniques. 

\subsection{Analysis of \texorpdfstring{$\RBF$}{RBF}}

We now show that when the thresholds are not too high, $\RBF$ outputs an allocation
where each agent $i$ gets a bundle of value at least $\tau_i$.

\begin{lemma}[follows from \cite{garg2021improved}]
\label{thm:rbf-gt}
Let $A$ be the allocation output by $\RBF(\Ical, T)$. For all $i \in [n]$,
if $\tau_i = \frac{3}{4} + \frac{1}{12n}$, then $v_i(A_i) \ge \tau_i$.
\end{lemma}

Suppose $p$ reductions happen during $\RBF$, where the $j\Th$ reduction is a type-$r_j$ reduction.
Call $R \defeq (r_1, \ldots, r_p)$ the \emph{reduction sequence}.
To analyze $\RBF$, we must understand the structure of $R$.

\begin{lemma}
\label{thm:redn-seq}
The reduction sequence $R$ for $\RBF$ is captured by the following regular expression: $(1^*2^*4^*)(32^*4^*)^*$.
Equivalently, if $n_3$ type-3 reductions happened in $\RBF$, then we can partition $R$ into
$n_3+1$ contiguous sub-lists $(R_0, \ldots, R_{n_3})$, such that all of the following hold:
\begin{enumerate}
\item $R_0$ is a sequence of zero or more 1s, followed by zero or more 2s, followed by zero or more 4s.
    ($R_0$ can be empty.)
\item For $i \in [n_3]$, the first element in $R_i$ is 3, followed by zero or more 2s, followed by zero or more 4s.
\end{enumerate}
\end{lemma}
\begin{proof}
All type-1 reductions happen together in the beginning.
This is because a different kind of reduction can only happen if a type-1 reduction is not possible.
If a type-1 reduction is not possible, then for each good $g$ and each remaining agent $i$,
$v_i(g) < \tau_i$. Hence, a type-1 reduction will never happen in the future.

Call $R_i$ the $i\Th$ \emph{phase} of \ $\RBF$.
Within each phase, all type-2 reductions happen before all type-4 reductions.
This is because a type-4 reduction happens only if a type-2 reduction is not possible.
After a type-4 reduction, the value of $S_2$ remains the same,
so a type-2 reduction will never occur again in the same phase.
\end{proof}

\begin{lemma}
\label{thm:m-ge-2n}
In $\RBF$, after all type-1 reductions have happened, we have $|M| \ge 2|N|$
throughout the execution of the algorithm.
\end{lemma}
\begin{proof}
Let $\Ical \defeq ([n], [m], \Vcal)$ be the initial fair division instance (the input to $\RBF$).

By \cref{thm:redn-seq}, all type-1 reductions happen before all other reductions.
Suppose $k$ type-1 valid reductions happened. Immediately after that,
$|M| = m - k$ and $|N| = n - k$. If $|M| < 2|N|$, then $m < 2n-k$.
For any agent $i \in [n]$, consider her partition $P \defeq (P_1, \ldots, P_n)$ for $\Ical$.
Since $m < 2n-k$, at least $k+1$ singleton bundles in $P$ exist.
Since $\Ical$ is normalized, the goods in those singleton bundles have a value of $1$ for $i$.
But at least $k+1$ type-1 reductions should have happened, which is a contradiction.
Hence, immediately after type-1 reductions, we have $|M| \ge 2|N|$.

Type-3 and type-4 reductions only happen when $|M| \ge 2|N|+1$,
so they preserve the invariant $|M| \ge 2|N|$.
\end{proof}

\begin{lemma}
\label{thm:redn-vs-bag}
In $\RBF$, after all reductions have happened, let $M_f$ and $N_f$ be
the set of remaining goods and remaining agents, respectively.
(By \cref{thm:m-ge-2n}, $|M_f| \ge 2|N_f|$.)
Then for any good $g$ in a type-$k$ reduction bundle, we have
\begin{enumerate}
\item $g > \ordSt(M_f, |N_f|)$ if $k = 2$.
\item $g > \ordSt(M_f, 2|N_f|)$ if $k = 3$.
\end{enumerate}
\end{lemma}
\begin{proof}
Suppose $p$ reductions happen during $\RBF$. For any $i \in [p]$,
let $S_i$ be the $i\Th$ reduction bundle and let $k_i$ be the type of the $i\Th$ reduction.
Let $M_i$ and $N_i$ be the set of remaining goods and remaining agents, respectively,
after $i$ reductions have happened. (Hence, $(N_0, M_0, \Vcal)$ is the original instance.)

Define proposition $P(i)$ as: For any $j \in [i]$ and any $g \in S_i$,
$g > \ordSt(M_i, |N_i|)$ if $k = 2$ and $g > \ordSt(M_i, 2|N_i|)$ if $k = 3$.

We prove $P(i)$ for all $i \in \{0\} \cup [p]$ using induction:
The base case $i = 0$ is trivial since no reductions have happened.
Now for $i \in [p]$, we want to prove $P(i-1) \implies P(i)$.

\textbf{Case 1}: $k_i = 1$:
\\ By \cref{thm:redn-seq}, no type-2 or type-3 reductions happened among the first $i$ reductions.
Hence, $P(i)$ holds trivially.

For the remaining cases, $\ordSt(M_i, |N_i|)$ and $\ordSt(M_i, 2|N_i|)$ are well-defined
because of \cref{thm:m-ge-2n}.

\textbf{Case 2}: $k_i = 2$:
\\ Then $S_i = \ordSt(M_{i-1}, \{|N_{i-1}|, |N_{i-1}|+1\})$,
$\ordSt(M_i, |N_i|) = \ordSt(M_{i-1}, |N_{i-1}|-1)$,
and $\ordSt(M_i, 2|N_i|) = \ordSt(M_{i-1}, 2|N_{i-1}|)$.
Hence, $P(i)$ holds.

\textbf{Case 3}: $k_i = 3$:
\\ Then $S_i = \ordSt(M_{i-1}, \{2|N_{i-1}|-1, 2|N_{i-1}|, 2|N_{i-1}|+1\})$,
$\ordSt(M_i, |N_i|) = \ordSt(M_{i-1}, |N_{i-1}|-1)$,
and $\ordSt(M_i, 2|N_i|) = \ordSt(M_{i-1}, 2|N_{i-1}|-2)$.
Hence, $P(i)$ holds.

\textbf{Case 4}: $k_i = 4$:
\\ Then $S_i = \ordSt(M_{i-1}, \{1, 2|N_{i-1}|+1\})$,
$\ordSt(M_i, 2|N_i|) = \ordSt(M_{i-1}, 2|N_{i-1}|-1)$,
and $\ordSt(M_i, |N_i|) = \ordSt(M_{i-1}, |N_{i-1}|)$.
Hence, $P(i)$ holds.

By mathematical induction, we get that $P(p)$ holds.
\end{proof}

\begin{lemma}[Lemma 6 in \cite{akrami2023simplification}]
\label{thm:bag-top-le-third}
In an ordered and normalized fair division instance, for all $i,k \in [n]$,
if $v_i(\{k, 2n+1-k\}) > 1$, then $v_i(2n+1-k) \le 1/3$ and $v_i(k) > 2/3$.
\end{lemma}

\begin{lemma}
\label{thm:rbf-1}
Let $A$ be the allocation output by $\RBF(\Ical, T)$.
For all $i \in [n]$, if $\tau_i = 2n/(2n+i-1)$, then $v_i(A_i) \ge \tau_i$.
\end{lemma}
\begin{proof}
Suppose for some $i \in [n]$, we have $\tau_i = 2n/(2n+i-1)$ and $v_i(A_i) < \tau_i$.
We will try to upper-bound the total value agent $i$ has for all the bundles.
To do this, we associate a \emph{charge} $c(A_j)$ with each bundle $A_j$ for $j \in [n]$
such that all of these properties hold:
\begin{enumerate}
\item The total charge of all bundles equals the total value of all bundles.
\item The charge of each bundle is less than $3\tau_i/2$.
\item If a bundle has value less than $\tau_i$ to agent $i$, then it is charged less than $\tau_i$.
\end{enumerate}

The total number of reductions that happen is at most $n-1$
(otherwise, $i$ would have gotten a bundle in some reduction).
In the instance after all reductions have happened,
let $M_f$ and $N_f$ be the set of remaining goods and remaining agents, respectively.
Let $z' \defeq v_i(\ordSt(M_f, |N_f|))$ and $z \defeq v_i(\ordSt(M_f, |N_f|+1))$.
Then $z' + z < \tau_i$ and $z' \ge z$, so $z < \tau_i/2$.

We now upper-bound the charge of each bundle in $A$.
To do this we consider three cases.

\textbf{Case 1}: type-1 and type-4 reduction bundles:
\\ Let $S$ be a type-$k$ reduction bundle, for $k \in \{1, 4\}$. Let $c(S) \defeq v_i(S)$.
Let $M$ and $N$ be the set of remaining goods and remaining agents immediately before the reduction.
Then $|N| \ge 2$, since at most $n-1$ valid reductions happen.
\begin{enumerate}
\item $k = 1$: then $v_i(S) \le v_i(1) \le 1$, since the instance is normalized.
\item $k = 4$: then a type-3 reduction was not possible. Hence,
    $v_i(S) = v_i(\ordSt(M, \{1, 2|N|+1\})) < \tau_i + \tau_i/3 = 4\tau_i/3$.
\end{enumerate}

\textbf{Case 2}: type-2 and type-3 reduction bundles:
\\ By \cref{thm:redn-seq}, we can partition the sequence of
reduction events into sub-lists $(R_0, R_1, \ldots, R_{n_3})$.
Let $S$ be a type-2 reduction bundle in $R_0$. Then $S = \{j, 2n+1-j\}$ for some $j$.
If $v_i(S) > 1$, then by \cref{thm:bag-top-le-third}, we get that $v_i(2n+1-j) \le 1/3$.
Also, $v_i(j) < \tau_i$, since no type-1 reduction was possible.
Hence, $v_i(S) < \tau_i + 1/3$ if $v_i(S) > 1$.
Since $\tau_i = 2n/(2n+i-1) \ge 2n/(3n-1) > 2/3$, we get
$1 < 3\tau_i/2$ and $\tau_i + 1/3 < 3\tau_i/2$.
Hence, $v_i(S) < 3\tau_i/2$. Let $c_i(S) \defeq v_i(S)$.

Consider the sub-list $R_t$ for some $t > 0$.
In the instance immediately before the first reduction in $R_t$,
let $M'$ and $N'$ be the set of remaining goods and agents, respectively.
Let $S_j$ be the $j\Th$ type-2 reduction bundle in $R_t$,
and let $T$ be the type-3 reduction bundle in $R_t$.
\begin{enumerate}
\item By \cref{thm:redn-vs-bag}, $\ordSt(M', 2|N'|-1) > \ordSt(M_f, 2|N_f|)$.
    Hence, $v_i(T) \le 3z < 3\tau_i/2$.
\item $S_1 = \ordSt(M', \{|N'|-1, |N'|\})$. By \cref{thm:redn-vs-bag},
    $\ordSt(M', |N'|-1) > \ordSt(M_f, |N_f|)$. Hence, $v_i(S_1) \le 2z' < 2(\tau_i - z)$.
\item For $j \ge 2$, we have $S_j = \ordSt(M', \{|N'|-j, |N'|+(j-1)\})$.
    A type-2 reduction was not possible at the beginning of $R_t$,
    so $v_i(\ordSt(M', |N'|+1)) < \tau_i/2$. Hence, $v_i(S_j) < 3\tau_i/2$.
    Let $c(S_j) \defeq v_i(S_j)$.
\end{enumerate}
If no type-2 reduction happened in $R_t$, let $c(T) \defeq v_i(T)$.

Now, suppose at least one type-2 reduction happened in $R_t$.
Let $\zeta \defeq v_i(T) + v_i(S_1) < 3z + 2(\tau_i - z) < 5\tau_i/2$.
Define $c(T)$ and $c(S_1)$ as follows:
\[ (c(T), c(S_1)) = \begin{cases}
(\frac{3}{5}\zeta, \frac{2}{5}\zeta)
    & \textrm{ if } v_i(T) > \frac{3}{5}\zeta
\\ (\frac{2}{5}\zeta, \frac{3}{5}\zeta)
    & \textrm{ if } v_i(S_1) > \frac{3}{5}\zeta
\\ (v_i(T), v_i(S_1)) & \textrm{ otherwise}
\end{cases}. \]
Let $x \defeq \max(c(T), c(S_1))$ and $y \defeq \min(c(T), c(S_1))$.
Then $x \le \frac{3}{5}\zeta < 3\tau_i/2$.
Furthermore, if $v_i(S_1) < \tau_i$, then $c(S_1) < \tau_i$.
This is because if $c(S_1) \neq v_i(S_1)$, then $c(S_1) = \frac{2}{5}\zeta < \tau_i$.
Similarly, $v_i(T) < \tau_i \implies c(T) < \tau_i$.

\textbf{Case 3}: bag filling bundles:
\\ For a bag filling bundle, its initial value is less than $\tau_i + z < 3\tau_i/2$.
If it received some goods during bag filling, then its value to $i$ before the last good
was added is less than $\tau_i$. Since no type-3 reduction is possible,
each good added to a bag has a value less than $\tau_i/3$ to agent $i$.
Hence, the bag's value is less than $4\tau_i/3$ to agent $i$.
Charge a bag equal to its value.

Now, we combine all the cases. There are $n$ bundles in the allocation.
Since agent $i$ did not get a bundle of value $\ge \tau_i$,
and she has higher priority than $n-i$ other agents, we get that
each such bundle has value less than $\tau_i$ to agent $i$,
and hence has charge less than $\tau_i$.
The remaining $i-1$ agents get a bundle of charge less than $3\tau_i/2$. Hence,
\[ n = \sum_{j=1}^n v_i(A_j) = \sum_{j=1}^n c(A_j)
    < (i-1)(3\tau_i/2) + (n-i+1)\tau_i = (\tau_i/2)(2n+i-1). \]
Therefore, $\tau_i > 2n/(2n+i-1)$. This is a contradiction.
Hence, if $\tau_i = 2n/(2n+i-1)$, then $v_i(A_i) \ge \tau_i$.
\end{proof}

\begin{theorem}
\label{thm:rbf}
Let $\Ical$ be a fair division instance and $T \defeq (\tau_1, \ldots, \tau_n)$ where for all $i \in [n]$,
\[ \tau_i \defeq \max\left(\frac{2n}{2n+i-1}, \frac{3}{4} + \frac{1}{12n}\right). \]
Let $A$ be the allocation output by $\RBF(\Ical, T)$. Then $A$ is a $T$-MMS allocation,
i.e., for all $i \in [n]$, we have $v_i(A_i) \ge \tau_i$.
\end{theorem}
\begin{proof}
Follows from \cref{thm:rbf-1,thm:rbf-gt}.
\end{proof}

\subsection{Best-of-Both-Worlds Fairness}
\label{sec:algo1:bobw}

In the best-of-both-worlds setting, the aim is to output a distribution (ideally of small support)
over allocations that is both ex-ante fair and ex-post fair.
We first show what guarantees we can get for any algorithm by assigning priority ranks randomly.

\begin{lemma}
\label{thm:cyclic-perm}
Let $T \defeq (\tau_1, \ldots, \tau_n)$ be a list of thresholds and let $\Acal$ be
an algorithm that outputs a $T$-MMS allocation for every fair division instance having $n$ agents.
For $k \in [n]$, let $A^{(k)}$ be the allocation output by algorithm $\Acal$
when agent $i$ has priority rank
\[ \begin{cases}
i + k - 1 & \textrm{ if } i + k - 1 \le n
\\ i + k - 1 - n & \textrm{ otherwise}
\end{cases}. \]
If we sample $k$ uniformly randomly from $[n]$, then $A^{(k)}$ is
ex-post $\tau_n$-MMS and ex-ante $(\frac{1}{n}\sum_{i=1}^n \tau_i)$-MMS.
\end{lemma}
\begin{proof}
For each agent $i$, her priority rank is distributed uniformly over $[n]$.
\end{proof}

Note that in \cref{thm:cyclic-perm}, if $\Acal$ is deterministic, then instead of outputting
just a random allocation, we can output the entire distribution by finding $A^{(k)}$ for all $k \in [n]$.
This is useful if the agents want to verify that the distribution is indeed ex-ante fair.

We now show that by randomizing priority ranks as in \cref{thm:cyclic-perm}
and picking thresholds as in \cref{thm:rbf}, $\RBF$ is
ex-post $(\frac{3}{4} + \frac{1}{12n})$-MMS and ex-ante $(0.82536 + \frac{1}{36n})$-MMS.

\begin{restatable}{lemma}{rthmRbfAvg}
\label{thm:rbf-avg}
For all $n \ge 1$, we get
\[ \gamma \defeq \frac{1}{n}\sum_{i=1}^n \max\left(\frac{2n}{2n+i-1}, \frac{3}{4} + \frac{1}{12n}\right)
    \ge 2\ln\left(\frac{4}{3}\right) + \frac{1}{4} + \frac{1}{36n}. \]
\end{restatable}
\begin{proof}
(Proof deferred to \cref{sec:prio-extra:calculations}).
\end{proof}

\begin{theorem}
\label{thm:rbf-rand}
Let $\Ical$ be a fair division instance having $n$ agents, $T \defeq (\tau_1, \ldots, \tau_n)$, and
$\tau_i \defeq \max\big(\frac{2n}{2n+i-1},\allowbreak \frac{3}{4} + \frac{1}{12n}\big)$ for all $i \in [n]$.
Sample $k$ uniformly randomly from $[n]$, and let agent $i$'s priority rank be
$i + k - 1$ if $i + k - 1 \le n$ and $i + k - 1 - n$ otherwise.
Then the allocation output by $\RBF(\Ical, T)$ is ex-post $(\frac{3}{4} + \frac{1}{12n})$-MMS
and ex-ante $\beta$-MMS, where
$\beta \defeq 2\ln\left(\frac{4}{3}\right) + \frac{1}{4} + \frac{1}{36n}
\approx 0.82536 + \frac{1}{36n}$.
\end{theorem}
\begin{proof}
Follows from \cref{thm:rbf,thm:cyclic-perm,thm:rbf-avg}.
\end{proof}

%% file: hard1.tex
\subsection{Upper-Bounds on Thresholds}
\label{sec:hard1}

We show that if $\RBF(\Ical, T)$'s output is $T$-MMS
for every fair division instance $\Ical$ with $n$ agents,
where $T \defeq (\tau_1, \ldots, \tau_n)$ is a list of thresholds and $\tau_n > 0$,
then $\tau_i$ must be at most $3n/(3n+i-2)$ for all $i \in \{3, \ldots, n\}$.
We also show how this limits the best ex-ante fairness achievable by $\RBF(\Ical, T)$.

In particular, by setting $i = n$, we get that $\tau_n \le 3n/(4n-2)$,
i.e., $\RBF$'s output cannot be better than $3n/(4n-2)$-MMS,
which was also proved in Section 5 of \cite{akrami2023simplification}.
Hence, we generalize their tight example to the setting with agent priority ranks.

Let $n \ge 3$. For $i \in \{3, \ldots, n\}$, let
$\delta_i \defeq 1/(3n+i-2)$, $\alpha_i \defeq 3n\delta_i$,
and $u_i: \mathbb{Z}_{>0} \to \mathbb{R}_{\ge 0}$, where
\[ u_i(j) \defeq \begin{cases}
(2n-\ceil{j/2})\delta_i & \textrm{ if } j \le 2n
\\ n\delta_i & \textrm{ if } 2n < j < 2n+i
\\ 0 & \textrm{ if } j \ge 2n+i
\end{cases}. \]

\begin{observation}
$3/4 < \alpha_n < \ldots < \alpha_3 < 1$.
\end{observation}

\begin{lemma}
\label{thm:hard1-props}
For $i \in \{3, \ldots, n\}$, $u_i$ satisfies all of these properties:
\begin{enumerate}
\item \label{item:hard1-props:ord-norm}$u_i$ is ordered and normalized.
\item \label{item:hard1-props:irred}$\max(u_i(\{1\}), u_i(\{n, n+1\}), u_i(\{2n-1, 2n, 2n+1\}), u_i(\{1, 2n+1\})) \le \alpha_i$.
\item \label{item:hard1-props:bag-fill}For each $j \in [n]$, we have
    $u_i(\{j, 2n+1-j\}) \le \alpha_i$ and $u_i(\{j, 2n+1-j, 2n+i-1\}) \ge 1$.
\end{enumerate}
\end{lemma}
\begin{proof}
$u_i(2n) = n\delta_i = u_i(2n+1) > 0$, so $u_i$ is ordered.

Let $k \defeq n-i+1$. Let $M \defeq (M_1, \ldots, M_n)$ where
\[ M_j \defeq \begin{cases}
\{j, 2k+1-j\} & \textrm{ if } 1 \le j \le k
\\ \{k+j, 2n+3k+1-j, 2n+j-k\} & \textrm{ if } k+1 \le j \le n
\end{cases}. \]
Then we can verify that $u_i(M_j) = 1$ for all $j \in [n]$.

$u_i(\{1, 2n+1\}) = u_i(\{n, n+1\}) = (3n-1)\delta_i \le \alpha_i$.
$u_i(\{2n-1, 2n, 2n+1\}) = 3n\delta_i = \alpha_i$.

$u_i(\{j, 2n+1-j\}) = (3n-1)\delta_i \le \alpha_i$.
$u_i(\{j, 2n+1-j, 2n+i-1\}) = (4n-1)\delta_i \ge 1$.
\end{proof}

For any $\eps \in (0, 1)$ such that $1/\eps \in \mathbb{Z}_{>0}$,
define valuation function $w_{\eps}: [n/\eps] \to \mathbb{R}_{\ge 0}$ as $w_{\eps}(j) = \eps$ for all $j$.
Note that $w_{\eps}$ is ordered and normalized.

\begin{theorem}
\label{thm:hard1}
Let $T \defeq (\tau_1, \ldots, \tau_n)$, where $\tau_n > 0$.
If $\RBF(\Ical, T)$ outputs a $T$-MMS allocation for every fair division instance $\Ical$
with $n$ agents, then $\tau_i \le \alpha_i$ for all $i \in \{3, \ldots, n\}$,
where $\alpha_i \defeq 3n/(3n+i-2)$.
\end{theorem}
\begin{proof}
Assume \wLoG{} that each agent $i$ has priority rank $i$.
Suppose $\RBF(\Ical, T)$ outputs a $T$-MMS allocation
for every fair division instance $\Ical$ with $n$ agents,
but $\tau_i > \alpha_i$ for some $i \in \{3, \ldots, n\}$.

Pick $\eps \in (0, \tau_n/3)$ such that $1/\eps \in \mathbb{Z}_{>0}$.
Consider a fair division instance $([n], [n/\eps], v)$, where $v_1 = \ldots = v_i = u_i$
and $v_{i+1} = \ldots = v_n = w_{\eps}$.
By \cref{thm:hard1-props}.\ref{item:hard1-props:irred}, no reduction events take place during $\RBF$.
Hence, $\RBF$ simply delegates everything to $\bagFill$.

By \cref{thm:hard1-props}.\ref{item:hard1-props:bag-fill}, all bags initially have value
at most $\alpha_i$ for the first $i$ agents, and value $2\eps$ for the remaining agents.
Since $\alpha_i < \tau_i$ and $2\eps < \tau_n$, no agent wants any bag initially.

For the first $i$ agents, only goods in $[2n+i-1]$ have positive value.
Since $\RBF$ outputs a $T$-MMS allocation, each agent in $[i]$
receives a bag with at least one good from $[2n+i-1] \setminus [2n]$.
But this is impossible since there are only $i-1$ goods in $[2n+i-1] \setminus [2n]$.
This is a contradiction. Hence, $\tau_i \le \alpha_i$.
\end{proof}

Our approach in \cref{sec:algo1:bobw} to get best-of-both-worlds fairness is via \cref{thm:cyclic-perm},
i.e., we first find a list $T \defeq (\tau_1, \ldots, \tau_n)$ of thresholds such that
$\RBF(\Ical, T)$ outputs a $T$-MMS allocation for every fair division instance $\Ical$ with $n$ agents,
and then we run $\RBF$ with random priority ranks.
Then we get ex-post $\tau_n$-MMS and ex-ante $\frac{1}{n}(\sum_{i=1}^n \tau_i)$-MMS.
We now show the limitations of this approach.

By \cref{thm:hard1}, if we want any ex-post MMS guarantee for the output of \cref{algo:rbf},
the best ex-ante MMS guarantee we can get is
\[ \frac{1}{n}\left(2 + \sum_{i=3}^n \frac{3n}{3n+i-2}\right). \]
We now show an upper-bound on this quantity.

\begin{restatable}{lemma}{rthmHardIExa}
\label{thm:hard1-avg}
For all $n \ge 2$,
\[ \frac{1}{n}\left(2 + \sum_{i=3}^n \frac{3n}{3n+i-2}\right)
    \le 3\ln\left(\frac{4}{3}\right) + \frac{1}{2n} \approx 0.86305 + \frac{1}{2n}. \]
\end{restatable}
\begin{proof}
(Proof deferred to \cref{sec:prio-extra:calculations}).
\end{proof}

\begin{theorem}
\label{thm:hard1-exa}
If we fix the number of agents to $n$, the ex-ante MMS guarantee obtainable from $\RBF$
by picking priority ranks randomly (as in \cref{thm:cyclic-perm})
is not better than $\beta$-MMS, where
$\beta \defeq 3\ln(\frac{4}{3}) + \frac{1}{2n} \approx 0.86305 + \frac{1}{2n}$,
and the ex-post MMS guarantee is not better than $\gamma$-MMS,
where $\gamma \defeq 3n/(4n-2) = \frac{3}{4} + \frac{3}{8n-4}$.
\end{theorem}
\begin{proof}
Follows from \cref{thm:hard1,thm:cyclic-perm,thm:hard1-avg}.
\end{proof}

%% file: prelim-extra-arXiv.tex
\section{Missing Proofs of Section~\ref{sec:prelim}}
\label{sec:prelims-extra}

\ordNorm*
\begin{proof}
    Let $\ins$ be an arbitrary instance. We create a $d$-normalized ordered instance $\mathcal{I}'' = (N, M, V'')$ such that from any $1$-out-of-$d$ MMS allocation for $\mathcal{I}''$, one can obtain a $1$-out-of-$d$ MMS allocation for the original instance $\mathcal{I}$.

    First of all, we can ignore all agents $i$ with $\MMS^d_i=0$ since no good needs to be allocated to them. Recall that for all $i \in N$, $P^i = (P^i_1, \ldots, P^i_d)$ is a $d$-MMS partition of agent $i$. For all $i \in N$ and $g \in M$, we define $v'_{i,g} = v_i(g)/v_i(P^i_j)$ where $j$ is such that $g \in P^i_j$. Now for all $i \in N$, let $v'_i: 2^M \rightarrow \mathbb{R}_{\geq 0}$ be defined as an additive function such that $v'_i(S)=\sum_{g \in S} v'_{i,g}$.  Note that $v'_{i,g} \leq v_i(g)/\MMS^d_i(M)$ for all $g \in M$ and thus,
    \begin{align}
        v_i(S) \geq v'_i(S) \cdot \MMS^d_i(M). \label{ineq-apx}
    \end{align}
    Since $v'_i(P^i_j)=1$ for all $i \in N$ and  $j \in [d]$, $\mathcal{I}' = (N, M, V')$ is a $d$-normalized instance. If a $1$-out-of-$d$ MMS allocation exists for $\mathcal{I}$, let $X$ be one such allocation. By Inequality \eqref{ineq-apx}, $v_i(X_i) \geq v'_i(X_i) \cdot \MMS^d_i(M) \geq \MMS^d_i(M)$. Thus, every allocation that is $1$-out-of-$d$ MMS for $\mathcal{I}'$ is $1$-out-of-$d$ MMS for $\mathcal{I}$ as well.
    For all agents $i$ and $g \in [m]$, let $v''_{i,g}$ be the $g$-th number in the multiset of $\{v_i(1), \ldots, v_i(m)\}$. Let $v''_i: 2^M \rightarrow \mathbb{R}_{\geq 0}$ be defined as an additive function such that $v''_i(S)=\sum_{g \in S} v''_{i,g}$. Let $\mathcal{I''} = (N,M,V'')$. Note that $\mathcal{I}''$ is ordered and $d$-normalized.
    Barman and Krishnamurthy \cite{barman2020approximation} proved that for any allocation $X$ in $\mathcal{I''}$, there exists and allocation $Y$ in $\mathcal{I'}$ such that $v'_i(Y_i) \geq v''_i(X_i)$. Therefore, from any $1$-out-of-$d$ MMS allocation in $\mathcal{I}''$, one can obtain a $1$-out-of-$d$ MMS allocation in $\mathcal{I'}$ and as already shown before, it gives a $1$-out-of-$d$ MMS allocation for $\mathcal{I}$.
\end{proof}

%% file: prio-extra.tex

\input{hard2.tex}

\subsection{Bounding Sums using Integrals}
\label{sec:prio-extra:calculations}

\begin{lemma}
\label{thm:int-bound}
Let $f: [a, b] \to \mathbb{R}_{\ge 0}$ be a non-increasing function
and let $F$ be its antiderivative. Then
\[ f(b) + \int_a^b f(x)dx \le \sum_{i=a}^b f(i) \le f(a) + \int_a^b f(x)dx. \]
\end{lemma}

\rthmRbfAvg*
\begin{proof}
Define $f: [0, n-1] \to \mathbb{R}_{\ge 0}$ as
\[ f(x) \defeq \max\left(\frac{2n}{2n+x},\,\frac{3}{4} + \frac{1}{12n}\right). \]
Let $\beta \defeq 2n(3n-1)/(9n+1)$. Then
\[ \frac{2n}{2n+x} \ge \frac{3}{4} + \frac{1}{12n} \iff x \le \beta. \]
Hence,
\[ f(x) \defeq \begin{cases}
\displaystyle \frac{2n}{2n+x} & \textrm{ if } x \in [0, \beta)
\\[1em] \displaystyle \frac{3}{4} + \frac{1}{12n} & \textrm{ if } x \in (\beta, n-1]
\end{cases}. \]
Therefore,
\newcommand*{\longNegSpace}{\!\!\!\!\!\!}
\begin{align*}
\gamma &= \frac{1}{n}\sum_{i=0}^{n-1} f(i)
    \ge \frac{1}{n}\left(\frac{3}{4} + \frac{1}{12n} + \int_0^{n-1}\longNegSpace f(x)dx\right)
    \tag{by \cref{thm:int-bound}}
\\ &= \frac{1}{n}\left(2n\ln\left(\frac{2n+\beta}{2n}\right)
    + \left(\frac{3}{4} + \frac{1}{12n}\right)(n-\beta)\right)
\\ &= 2\ln\left(\frac{4}{3}\right) + \frac{1}{4} + \frac{1}{4n} - 2\ln\left(1 + \frac{1}{9n}\right)
\\ &\ge 2\ln\left(\frac{4}{3}\right) + \frac{1}{4} + \frac{1}{4n} - \frac{2}{9n}
    \tag{since $\ln(1+x) \le x$ for all $x > -1$}
\\ &= 2\ln\left(\frac{4}{3}\right) + \frac{1}{4} + \frac{1}{36n}.
\qedhere
\end{align*}
\end{proof}

\rthmHardIExa*
\begin{proof}
\begin{align*}
& \frac{1}{n}\left(2 + \sum_{i=3}^n \frac{3n}{3n+i-2}\right)
= \frac{1}{n}\left(1 + \sum_{i=0}^{n-2} \frac{3n}{3n+i}\right)
\\ &\le \frac{1}{n}\left(2 + 3n\int_0^{n-2} \frac{dx}{3n+x}\right)
    \tag{by \cref{thm:int-bound}}
\\ &= \frac{2}{n} + 3\ln\left(\frac{4}{3}\right) + 3\ln\left(1 - \frac{1}{2n}\right)
\\ &\le \frac{2}{n} + 3\ln\left(\frac{4}{3}\right) - \frac{3}{2n}
    \tag{since $\ln(1-x) \le -x$ for all $x < 1$}
\\ &= 3\ln\left(\frac{4}{3}\right) + \frac{1}{2n}.
\qedhere
\end{align*}
\end{proof}

\rthmHardIIExa*
\begin{proof}
Define $f: [0, n-1] \to \mathbb{R}_{\ge 0}$ as
\[ f(x) \defeq \min\left(\frac{3n}{3n+x-1}, \max\left(\frac{5}{6}, 1 - \frac{x}{3n}\right)\right). \]
Then for $x \ge 0$,
\[ \frac{3n}{3n+x-1} - \left(1 - \frac{x}{3n}\right)
= \frac{x(x-1) + 3n}{3n(3n+i-2)} > 0. \]
\[ \frac{5}{6} \ge 1 - \frac{x}{3n} \iff x \ge \frac{n}{2}. \]
\[ \frac{5}{6} \le \frac{3n}{3n+x-1} \iff x \le \frac{3n}{5}+1. \]
Hence,
\[ f(x) = \begin{cases}
1 - x/3n & \textrm{ if } x \in [0, n/2]
\\ 5/6 & \textrm{ if } x \in \big(\frac{n}{2}, \frac{3n}{5}+1\big]
\\ \displaystyle \frac{3n}{3n+x-1} & \textrm{ if } x \in \big(\frac{3n}{5}+1, n-1\big]
\end{cases}. \]
Therefore,
\begin{align*}
& \frac{1}{n}\sum_{i=1}^n f(i-1)
\le \frac{1}{n}\left(f(0) + \int_0^{n-1} f(x)dx\right)
    \tag{by \cref{thm:int-bound}}
\\ &= \frac{1}{n}\left(1 + \int_0^{\frac{n}{2}}\left(1 - \frac{x}{3n}\right)dx
    + \int_{\frac{n}{2}}^{\frac{3n}{5}+1} \frac{5}{6}dx + \int_{\frac{3n}{5}+1}^{n-1} \frac{3n}{3n+x-1}dx\right)
\\ &= \frac{1}{n}\left(\frac{11}{6} + \frac{13}{24}n
    + 3n\ln\left(\frac{10}{9}\left(1 - \frac{1}{2n}\right)\right)\right)
\\ &\le \frac{13}{24} + 3\ln\left(\frac{10}{9}\right) + \frac{1}{3n}.
    \tag{since $\ln(1-x) \le -x$ for all $x < 1$}
\end{align*}
\end{proof}

%% file: hard2.tex
\section{Limitations of Oblivious Analysis}
\label{sec:hard2}

We start by showing an attempt to improve the upper bounds of \cref{sec:hard1}.
Although our attempt does not succeed, it reveals a shortcoming of existing
techniques (including ours) to analyze bag-filling algorithms.
We formally characterize this using a concept called \emph{obliviousness},
and study the limitations of oblivious techniques.

\subsection{Attempt to get a Hard Example}
\label{sec:hard2:ex}

Consider a fair division instance $\Ical \defeq ([n], [m], v)$ with $n \defeq 6$ agents,
where each agent $i$ has priority rank $i$. We will analyze $\RBF$ from the perspective of agent 4.
\Cref{sec:hard1} says that $\tau_4$ should be at most $0.9$ for $\RBF$ to always succeed.
We will try to improve this to $\tau_4 \le 5/6 = 0.8\overline{3}$.
According to agent 4, say the first 3 goods have value $5/6$, the next 9 goods of value $1/3$, and
the remaining $3t$ goods have value $\eps \defeq 1/6t$, for some $t \in \mathbb{Z}_{\ge 3}$.

$\Ical$ is ordered and normalized according to agent 4:
3 of her MMS bundles have 3 goods of value $1/3$ each, and each of the remaining bundles
has one good of value $5/6$ and $t$ goods of value $1/6t$.

Suppose $\tau_4 > 5/6 + 2\eps$. Then agent 4 does not want any bundle during reduction, since
$v_1(1) \le v_1(\{1, 2n+1\}) = 5/6 + \eps < \tau_4$, $v_1(\{n, n+1\}) = 2/3 < \tau_4$,
and $v_1(\{2n-1, 2n+1, 2n+1\}) = 2/3 + \eps < \tau_4$.
Suppose the other agents also do not want any of the above bundles.
Hence, no reduction events happen during $\RBF$, and $\bagFill$ starts.

According to agent 4, the first 3 bags have value $5/6 + 1/3 = 7/6$ each.
Suppose the first 3 agents also want these bags.
Since they have a smaller priority rank, these bags go to the first 3 agents.
Agent 4 does not like any of the remaining bags, since they have value $2/3$.
Suppose agents 5 and 6 also do not like these bags. Hence, we start adding goods to these bags.
Since all goods in $[m] \setminus [2n]$ have value $\eps$ to agent 4, agent 4 would want
a bag only if at least $t+3$ goods are added to it. Suppose agents 5 and 6 like a bag
iff it gets at least $t+2$ goods from $[m] \setminus [2n]$.
Since there are only $3t$ goods of value $\eps$, agent 4 does not get a bag of value $\ge \tau_4$.

At first glance, this seems to prove that $\RBF$ cannot give agent 4 more than $5/6$-MMS.
But there is a catch: we have not mentioned what the valuations and thresholds of the other agents are.
Are there valuations and thresholds for which $\RBF$ behaves as above?

When all agents have the same valuation function as agent 4, and agents 5 and 6 have thresholds
in the range $(5/6+\eps, 5/6+2\eps]$, then $\RBF$ indeed behaves as described above.
However, this threshold range is unreasonable since \cref{sec:hard1} says that
$\tau_6 \le 9/11 < 5/6$ for $\RBF$ to succeed.
Are there (ordered and normalized) valuations and reasonable thresholds for the agents such that
$\RBF$ behaves as described above? We could not resolve this question.

\subsection{Obliviousness}
\label{sec:hard2:obl}

Let us think of fair division algorithms as interactive protocols,
where the algorithm repeatedly makes value queries to agents, i.e.,
it finds out $v_i(S)$ for some agent $i$ and some set $S$ of goods,
instead of just querying the value of each good for each agent upfront.
$\RBF$ can be interpreted in this way: during the reduction phase,
it queries $v_i(S)$ for $S \in \{\{1\}, \{n, n+1\}, \{2n-1, 2n, 2n+1\}, \{1, 2n+1\}\}$,
and during the bag-filling phase, it repeatedly queries the value of some bag.

Let $\Acal$ be an algorithm for fair division, let $i^*$ be an agent,
and let $T \defeq (\tau_1, \ldots, \tau_n)$ be a list of thresholds.
Generally, when we analyze $\Acal$, we ask the following question:
\textsl{Does $\Acal$'s output give each agent $i$ at least $\tau_i$ times her MMS,
given that all agents answer value queries truthfully?}

An \emph{oblivious analysis} of $\Acal$, on the other hand, asks the following question:
\textsl{For each $i \in [n]$, does $\Acal$'s output give agent $i$ at least $\tau_i$ times her MMS
if agent $i$ answers value queries truthfully and
all other agents are free to answer value queries non-truthfully?}
Note that in an oblivious analysis, the other agents' responses are not required to be
consistent with any (additive, ordered, normalized) valuation function.
Also, we do not assume that the other agents are rational, i.e., they may answer value
queries arbitrarily, even if it hurts them.

In the oblivious setting, we wish to show that an agent $i$ can secure a certain
fraction of her MMS for any possible responses of the other agents.
This is a worst-case analysis, so imagine that the other $n-1$ agents are
controlled by an adversary (who knows $i$'s valuation function)
whose sole motive is to reduce the value of $i$'s bundle.
Furthermore, the oblivious setting is more demanding than a non-oblivious one, i.e.,
if $i$ is guaranteed $\alpha$ times her MMS in the oblivious setting,
then she also gets $\alpha$ times her MMS in the non-oblivious setting.

The analysis of bag-filling algorithms in the literature
\cite{akrami2023simplification,garg2021improved,Akrami2023BreakingT,garg2019approximating},
and our analysis in \cref{sec:algo1}, work in the oblivious setting.
In fact, it is unclear how to \emph{not} do an oblivious analysis.
In this sense, obliviousness is a hard-to-bypass barrier towards the design
of approximate MMS algorithms.

Furthermore, we can get tighter upper bounds on thresholds in the oblivious setting.
E.g., \cref{sec:hard2:ex} shows that for 6 agents,
we need $\tau_4 \le 5/6$ for $\RBF$ to succeed in the oblivious setting.
We now generalize this example to an arbitrary number of agents.

\begin{theorem}
\label{thm:hard2}
Let $T \defeq (\tau_1, \ldots, \tau_n)$ be a list of thresholds, where $n \ge 2$.
If $\RBF(\cdot, T)$ always outputs a $T$-MMS allocation in the oblivious setting,
then for each $i \in [n]$, we have
\[ \tau_i \le \max\left(\frac{5}{6}, 1-\frac{i-1}{3n}\right). \]
\end{theorem}
\begin{proof}
Let $k_1, k_2 \in [n]$ such that $k_1 + k_2 < i$ and $2k_1 + k_2 \le n$.
Let $\alpha \defeq 1 - k_1/3(n-k_2)$. Let $t \in \mathbb{Z}_{\ge 3}$
and $\eps \defeq 1/3t(n-k_2)$. Then $\alpha \in [5/6, 1-\eps)$.
Pick an arbitrary agent $i$. Assume \wLoG{} that $i \ge 2$ (otherwise theorem is trivially true).
According to agent $i$, let there be $k_1+k_2$ goods of value $\alpha$,
$2n-k_1-2k_2$ goods of value $1/3$, and $(n-k_1-k_2)^2t$ goods of value $\eps$.
Let $\tau_i > \alpha+2\eps$.

The instance is normalized for agent $i$, since the MMS partition is as follows:
\begin{enumerate}
\item $k_1$ bundles: one good of value $\alpha$ and $k_1t$ goods of value $\eps$.
\item $k_2$ bundles: one good of value $\alpha$ and one good of value $1-\alpha$.
\item $k_1$ bundles: 3 goods of value $1/3$.
\item $n-2k_1-k_2$ bundles: 2 goods of value $1/3$ and $1/3\eps$ goods of value $\eps$.
\end{enumerate}

Suppose no reduction events happen during $\RBF$.
This is possible because $v_i(\{n, n+1\}) = 2/3 < \tau_i$,
$v_i(\{2n-1, 2n, 2n+1\}) \le 2/3 + \eps < \tau_i$, and $v_i(\{1, 2n+1\}) = \alpha+\eps < \tau_i$.

Suppose $\RBF$ gives away $k_1+k_2$ bags having a good of value $\alpha$
and a good of value $1/3$ or $1-\alpha$ to the first $k_1 + k_2$ agents.
All the remaining bags initially have 2 goods of value $1/3$.
Suppose each of them receives at most $(n-k_1-k_2)t + 2$ goods of value $\eps$.
Then each of those bags has value at most $\alpha+2\eps$.
Hence, agent $i$ will not get a bag of value $\ge \tau_i$.

Hence, if $\tau_i > \alpha$, then for some large enough $t$,
$\RBF$ cannot give a bag of value $\ge \tau_i$ to agent $i$.
If $i \le n/2+1$, then set $k_1 = i-1$ and $k_2 = 0$ to get $\tau_i \le 1 - (i-1)/3n$.
Otherwise, set $k_1 = \floor{n/2}$ and $k_2 = n - 2k_1$ to get $\tau_i \le 5/6$. Hence,
\[ \tau_i \le \max\left(\frac{5}{6}, 1-\frac{i-1}{3n}\right).
\qedhere \]
\end{proof}

Our approach in \cref{sec:algo1:bobw} to get best-of-both-worlds fairness is by
running $\RBF$ with random priority ranks (as in \cref{thm:cyclic-perm}).
If $\RBF$ is $T$-MMS, this gives us ex-post $\tau_n$-MMS
and ex-ante $\frac{1}{n}(\sum_{i=1}^n \tau_i)$-MMS.
In the oblivious setting, by \cref{thm:hard1,thm:hard2},
the best ex-ante MMS guarantee we can get is
\[ \frac{1}{n}\sum_{i=1}^n \min\left(\frac{3n}{3n+i-2},
    \max\left(\frac{5}{6}, 1 - \frac{i-1}{3n}\right)\right). \]
We now show an upper-bound on this quantity.

\begin{restatable}{lemma}{rthmHardIIExa}
\label{thm:hard2-avg}
For all $n \ge 1$,
\[ \frac{1}{n}\sum_{i=1}^n \min\left(\frac{3n}{3n+i-2},
    \max\left(\frac{5}{6}, 1 - \frac{i-1}{3n}\right)\right)
    \le \frac{13}{24} + 3\ln\left(\frac{10}{9}\right) + \frac{1}{3n}. \]
\end{restatable}
\begin{proof}
(Proof deferred to \cref{sec:prio-extra:calculations}).
\end{proof}

\begin{theorem}
\label{thm:hard2-exa}
If we fix the number of agents to $n$, the ex-ante MMS guarantee obtainable from $\RBF$
in the oblivious setting by picking priority ranks randomly (as in \cref{thm:cyclic-perm})
is not better than $\beta$-MMS, where
$\beta \defeq \frac{13}{24} + 3\ln\left(\frac{10}{9}\right) + \frac{1}{3n}
\approx 0.8578 + \frac{1}{3n}$,
and the ex-post MMS guarantee is not better than $\gamma$-MMS,
where $\gamma \defeq 3n/(4n-2) = \frac{3}{4} + \frac{3}{8n-4}$.
\end{theorem}
\begin{proof}
Follows from \cref{thm:hard1,thm:hard2,thm:cyclic-perm,thm:hard1-avg}.
\end{proof}